\documentclass[12 pt]{article}
\usepackage{bbm}
\usepackage[longnamesfirst]{natbib}
\usepackage{amssymb}
\usepackage{amsmath}
\usepackage{amsthm}
\usepackage{setspace}
\usepackage[colorlinks]{hyperref}
\usepackage{graphicx}
\usepackage{caption}
\usepackage{enumitem}
\setlist{noitemsep,parsep=6pt,partopsep=0pt,topsep=0pt}
\usepackage[margin=1in]{geometry}
\usepackage{verbatim}
\usepackage{sectsty}
\usepackage[svgnames]{xcolor}
\usepackage{subfiles}
\usepackage{xargs}
\usepackage[titletoc,title]{appendix}
\usepackage{titlesec}
\usepackage[bottom]{footmisc}
\usepackage[colorinlistoftodos,textsize=tiny]{todonotes}
\usepackage[capitalize,nameinlink]{cleveref}
\usepackage[normalem]{ulem}
\usepackage[font={small,it},justification=justified]{caption}
\hypersetup{
   pdftitle={Convex Choice},
    pdfauthor={Navin Kartik, Andreas Kleiner},
    citecolor=DarkBlue,
    bookmarksnumbered=true,
    urlcolor=Indigo,linkcolor=DarkBlue
}
\usepackage{subcaption}
\theoremstyle{remark}
\newtheorem{remark}{Remark}

\theoremstyle{plain}

\newtheorem{conclusion}{Conclusion}

\newtheorem{lemma}{Lemma}

\newtheorem{proposition}{Proposition}

\theoremstyle{definition}
\newtheorem{definition}{Definition}
\newtheorem{example}{Example}
\AtEndEnvironment{example}{\examplemarker}

\newcommand{\examplemarker}{\hfill$\diamond$}
\newcommand{\marker}{\hfill\textsquare}
\renewcommand{\qedsymbol}{\textsquare}
\renewcommand{\epsilon}{\varepsilon}
\newcommand{\bernoulli}{\bar{u}}
\renewcommand{\qedsymbol}{Q.E.D.}

\newcommand{\citepos}[1]{\citeauthor{#1}'s (\citeyear{#1})} 
\newcommand{\citeauthorpos}[1]{\citeauthor{#1}'s} 
\newcommand{\R}{\mathbb{R}}

\def\Reals{\mathbb R}
\newcommand{\sopt}{S} 
\newcommand{\wopt}{W} 

\newcommand{\e}{\varepsilon}

\DeclareMathOperator*{\argmax}{\arg\max}
\DeclareMathOperator{\cl}{cl}
\DeclareMathOperator{\sign}{sign}
\DeclareMathOperator{\ri}{ri}
\DeclareMathOperator{\Span}{span}

\usetikzlibrary{patterns, decorations.pathreplacing, arrows.meta,calc,decorations.pathreplacing,calligraphy}

\let \savenumberline \numberline
\def \numberline#1{\savenumberline{#1.}}

\makeatletter
  \renewcommand\@seccntformat[1]{\csname the#1\endcsname.{\hskip.7em\relax}} 
\makeatother
\makeatletter
\renewenvironment{proof}[1][\proofname] {\par\pushQED{\qed}\normalfont\topsep6\p@\@plus6\p@\relax\trivlist\item[\hskip\labelsep\bfseries#1\@addpunct{.}]\ignorespaces}{\popQED\endtrivlist\@endpefalse}
\makeatother

\newcommand{\mailto}[1]{\href{mailto:#1}{\texttt{#1}}} 

\let\oldfootnote\footnote
\renewcommand\footnote[1]{\oldfootnote{\hspace{.5mm}#1}}

\titlespacing\section{0pt}{10pt plus 2pt minus 2pt}{4pt plus 2pt minus 2pt} 
\titlespacing\subsection{0pt}{6pt plus 2pt minus 2pt}{2pt plus 2pt minus 2pt} 
\titlespacing\subsubsection{0pt}{6pt plus 2pt minus 2pt}{0pt plus 2pt minus 2pt} 
\titlespacing{\paragraph}{%
  0pt}{
  0.5\baselineskip}{
  1em}

\newcommand{\appendixref}[1]{\hyperref[#1]{Appendix \ref{#1}}}

\usepackage{xcolor}
\usepackage{pgf}
\usepackage{tikz}
\usetikzlibrary{shapes,arrows,automata,fit}

\tikzstyle{info}=[circle,thick,draw=black,fill=black!25,minimum size=4mm]
\tikzstyle{uninfo}=[circle,thick,draw=black,fill=white,minimum size=4mm]
\tikzstyle{inforecog}=[circle,line width=1mm,draw=black!50,fill=black!25,minimum size=4mm]
\tikzstyle{uninforecog}=[circle,line width=1mm,draw=black!50,fill=white,minimum size=4mm]

\tikzstyle{traded}=[draw, line width=1mm]
\tikzstyle{recog}=[draw=black!50, line width=1mm]

\sectionfont{\color{DarkRed}}
\subsectionfont{\color{DarkRed}}
\subsubsectionfont{\color{DarkRed}}

\begin{document}

\begin{titlepage}

\title{Convex Choice\thanks{We thank Gerrit Bauch, Bo Chen, Laura Doval, Nima Haghpanah, SangMok Lee, Elliot Lipnowski, Michael Ostrovsky, Daniel Rappoport, Joel Sobel, Roland Strausz, Bruno Strulovici, Lina Uhe, Zhihui Wang, Cole Wittbrodt, Yangfan Zhou, and conference and seminar audiences for helpful comments. Kleiner 
acknowledges financial support from the German Research Foundation (DFG) through Germany’s Excellence Strategy - EXC 2047/2 - 390685813, EXC 2126/2-390838866 and the CRC TR-224 (Project B02).}}
\author{
Navin Kartik\thanks{Department of Economics, Yale University.  Email: \mailto{nkartik@gmail.com}.}
\and 
Andreas Kleiner\thanks{Department of Economics, University of Bonn.  Email:  \mailto{andykleiner@gmail.com}.}
}

\maketitle

\begin{abstract}
\noindent 
For multidimensional Euclidean type spaces, we study convex choice: from any choice set, the set of types that make the same choice is convex. We establish that, in a suitable sense, this property characterizes the sufficiency of local incentive constraints. Convex choice is also of interest more broadly, e.g., in cheap-talk games. We tie convex choice to a notion of directional single-crossing differences (DSCD). For an expected-utility agent choosing among lotteries, DSCD implies that preferences are either one-dimensional or must take the affine form that has been tractable in multidimensional mechanism design.

\end{abstract}
\thispagestyle{empty} 

\end{titlepage}

\setstretch{1.1}

\onehalfspacing

\section{Introduction}
\label{sec:intro}

A fundamental property in economic models with one-dimensional private information is the Spence-Mirrlees single-crossing property. Roughly equivalent to ``interval choice''---the set of types for which an action is optimal is an interval---the single-crossing property has many uses. 
For example, it reduces incentive compatibility (IC) to local IC. Local IC requires only that ``nearby'' types have no incentive to mimic each other. Sufficiency of local IC is central to the tractability, and hence success, of influential paradigms like signaling and mechanism design.

Our paper is concerned with multidimensional environments. Consider an agent with utility function $u(a,\theta)$, where $a\in A$ is an allocation or choice variable and $\theta \in \Theta$ is the agent's type or preference parameter.  Let $\Theta \subset \Reals^n$ be convex. We study a property we call \emph{convex choice}: for any choice set $B\subset A$ and any allocation $a\in B$, the set of types for which $a$ is (uniquely) optimal is convex.\footnote{Our notion relates the preferences of different types. It is distinct from ``convex preferences'' in decision theory, which refers to the following property of a single preference relation: if $a$ is preferred to $b$, then convex combinations of $a$ and $b$ (in some suitable sense) are also preferred to $b$.%
}

We view convex choice as compelling. There is a sense in which it characterizes the sufficiency of local IC. More precisely, absent indifferences, convex choice characterizes when IC between any two types is equivalent to local IC along the line segment between those two types (\autoref{prop:CC-IC}). The necessity of convex choice for such ``integration up'' of local IC suggests that absent convex choice, a mechanism-design model is unlikely to be tractable.\footnote{We also discuss how sufficiency of local IC on the full type space is more or less characterized by a weaker property we call \emph{connected choice} (\autoref{prop:connected-IC}); this property does not, however, deliver the tractability of integrating up local IC along line segments.} Convex choice is relevant for other reasons as well. For instance, it delivers an extension of \citepos{SaksYu05} result on the implementability of allocation rules in settings with transfers  (\autoref{prop:weakmon}). It also offers a natural generalization of ``interval equilibria'' in multidimensional versions of \citepos{CS82} famous cheap-talk model. 

\autoref{prop:CC-DSCD} shows that convex choice is essentially equivalent to a form of single crossing that we term \emph{directional single-crossing differences} (DSCD). Modulo details, DSCD requires that for any two (undominated) allocations, the set of types that prefer one to the other is defined by a half-space. Importantly, the orientation of the half-space can vary with the allocation pair. With $A,\Theta\subset \Reals^n_{\geq 0}$, leading examples of DSCD are the quadratic-loss utility $u(a,\theta)=-\Vert a-\theta \Vert^2$, where $\Vert \cdot \Vert$ is the Euclidean norm, and
the constant-elasticity-of-substitution utility $u(a,\theta)=\left(\sum_{i=1}^n (a_{i})^r\theta_{i} \right)^{s}$ with parameters $r,s>0$.

Studies of multidimensional mechanism design frequently assume the \emph{affine form} $u(a,\theta)=v(a) \cdot \theta + w(a)$ \citep[e.g.,][]{Armstrong:96,RC:98,MV:07,Kleiner:22}. Our approach delivers a novel perspective on this specification. Propositions  \ref{prop:DSCDstar} and \ref{prop:DSCDstar-strict-dominance} show that if the convex choice/DSCD requirement is strengthened to hold for expected utility over allocation lotteries, then---subject to some regularity conditions---either the setting is effectively one-dimensional, or preferences can be represented in the affine form.
In particular, this provides a justification for the common focus on quadratic-loss utility: when preferences depend only on the distance between allocation and type, essentially only quadratic loss with a weighted Euclidean norm has DSCD over lotteries.

One interpretation of these results is that when allocation lotteries are important, multidimensional mechanism design is unlikely to be generally tractable beyond the affine form. By ``tractable'', we mean finding a useful characterization of the set of IC mechanisms, which is the first step in the canonical Myersonian approach \citep{Myerson:81}, prior to optimization. Lotteries may be inescapable when stochastic mechanisms are allowed or when there are other agents
and Bayesian incentive constraints are considered.
On the flip side, a constructive interpretation of Propositions  \ref{prop:DSCDstar} and \ref{prop:DSCDstar-strict-dominance} is that they provide some guidance for when mechanism designers should focus on alternative, non-Myersonian, approaches; see, for example, \citet{HH21} in the context of monopoly bundling. In this regard, our \autoref{prop:weakmon} on implementability of allocation rules may be useful for future work.

\paragraph{Related Literature.}
As elaborated later, convex choice is closely related to notions in \cite{Grandmont:78} and, in contemporaneous work to ours, \citet{Rodriguez26} and \citet{WuWu2026}. In particular, \cite{Grandmont:78} has a similar characterization to our \autoref{prop:CC-DSCD}.
Beyond that, our emphasis and results are largely distinct:  \citet{Grandmont:78} and \citet{WuWu2026} focus on social choice, and \citet{Rodriguez26} on cheap talk; none of these papers relate convex choice to the sufficiency of local IC or implementability in mechanism design, or have any analog to our characterization of DSCD over lotteries.

Our work also relates to \citet{Carroll:12}, \citet{KLR:23}, and \citet{MM:88}. We detail these and some other connections subsequently.

\section{Convex Choice and Applications}
\label{sec:CC}
There is an agent with type $\theta \in \Theta \subset \R^n$, where $\Theta$ is convex. We write $\theta_{(i)}$ for the $i$-th coordinate of $\theta$, and $\theta>\theta'$ if $\theta_{(i)} \geq \ \theta'_{(i)}$ for all $i$  with strict inequality for some $i$. We denote the \emph{line segment} between types $\theta'$ and $\theta''$ by $$\ell(\theta',\theta''):=\left\{\theta: \exists \lambda\in[0,1] \text{ with } \theta=\lambda \theta'+(1-\lambda)
\theta''\right\}.$$
The agent takes an action, or chooses an allocation, $a\in A$, with preferences given by the utility function $u:A \times \Theta \to \Reals$. We denote the set of types for whom action $a$ is optimal from choice set $B\subset A$ either strictly or weakly by\footnote{Throughout, the ``$\subset$'' symbol means ``weak subset''.}
\[
\label{e:sopt}
\sopt_a^B:=\left\{\theta\in\Theta:
\{a\}=\argmax_{a'\in B}u(a',\theta)\right\}
\qquad\text{and}\qquad
\wopt_a^B:=\left\{\theta\in\Theta:
a\in\argmax_{a'\in B}u(a',\theta)\right\}.
\]

\begin{definition}
\label{def:CC}
$u$ has \emph{convex choice} if for all finite $B\subset A$ and $a\in B$, the set $\sopt_a^B$ is convex.
\end{definition}

In other words, convex choice requires that from any finite choice set, the set of types that find an action uniquely optimal is convex. It is equivalent to only consider all binary choice sets, as the intersection of convex sets is convex. Convex choice implies that the preferences of any type $\theta \in \ell(\theta',\theta'')$ are ``between'' those of $\theta'$ and $\theta''$, in the sense of \cite{Grandmont:78}.\footnote{\citeauthor{Grandmont:78}'s betweenness notion is stated for binary relations that need not be transitive; so, in his setting, choice from non-binary choice sets may not be well-defined. His notion also imposes some requirements concerning indifference. Ignoring indifferences and assuming transitivity, convex choice is equivalent to the preferences of any $\theta\in \ell(\theta',\theta'')$ being between those of $\theta'$ and $\theta''$.} Convex choice is also very similar to \citepos{Rodriguez26} ``convex-regular property''; the notions are equivalent absent indifferences.

Another related, but more substantially different, notion to convex choice is \citepos{KLR:23} ``interval choice''. They require that if $\theta'' > \theta' > \theta$ and some action is optimal for $\theta$ and $\theta''$, then it is also optimal for $\theta'$. (Those authors use weak optimality whereas we use strict optimality; this difference is minor.) Absent indifferences, convex and interval choice are equivalent when $\Theta \subset \Reals$, but they are incomparable when $\Theta$ is multidimensional (\autoref{eg:intervalvsconvex}). Convex sets are, of course, a generalization of one-dimensional intervals to multiple dimensions. We will discuss yet another alternative, connected sets, subsequently.

The interpretation of convex choice in any context depends on what the agent's type represents. Although we are formally studying a single-agent choice problem, we view the agent's type as having external economic meaning beyond just parameterizing the agent's preferences. For instance, the type could represent the agent's information, ability, or ideology on multiple dimensions, which is relevant to an observer, principal, or politician. We thus take the type space as given, and do not consider reparameterizing it.

\subsection{Sufficiency of Local Incentive Constraints}
\label{sec:connected}

A central question in mechanism design, and information economics more broadly, is when local incentive constraints ensure global ones. To define our notions of incentive compatibility, let $N_{\theta}\subset \Theta$ denote an open neighborhood---simply neighborhood hereafter---of type $\theta$ in the relative topology of $\Theta$. Throughout, we focus on direct mechanisms $\Theta \to A$; stochastic mechanisms are subsumed by taking $A$ to be a lottery space.

\begin{definition}
\label{def:IC}
A mechanism $m:\Theta \to A$ is
\begin{enumerate}
\item \emph{incentive compatible} (IC) if $u(m(\theta),\theta)\geq u(m(\theta'),\theta)$ for all $\theta,\theta'\in \Theta$;
\item \emph{locally IC} if for each $\theta\in \Theta$ there is a neighborhood $N_{\theta}\subset \Theta$ such that
\begin{equation}
\label{e:lIC}
\forall \theta'\in N_{\theta}: \; u(m(\theta),\theta)\geq u(m(\theta'),\theta) \text{ and } u(m(\theta'),\theta')\geq u(m(\theta),\theta').
\tag{LIC}
\end{equation}
\end{enumerate}
\end{definition}

We will also refer to (local) IC of mechanisms defined on a subset of the type space $\Theta'\subset \Theta$; that simply refers to \autoref{def:IC} with $\Theta'$ in place of $\Theta$.

IC is a standard and fundamental property. Our formulation of local IC follows \citet{Carroll:12}: for every type $\theta$, there is a set of nearby types for which (i) $\theta$ cannot profitably mimic those types and (ii) those types cannot profitably mimic $\theta$. Requirement (ii) here is due to the non-discrete type space, and cannot be dispensed with.\footnote{Here is an example. $\Theta=[0,1]$, $A=\{a',a''\}$, $u(a',\cdot)<u(a'',\cdot)$, and $m(\theta)=a'$ if and only if $\theta < 1/2$. For every type $\theta<1/2$, there is $\epsilon>0$ such that all types in $[0,\theta+\epsilon)$ receive the same allocation $a'$; every type $\theta\geq 1/2$ receives its preferred allocation $a''$. Hence, every $\theta$ has a neighborhood containing no type that $\theta$ would profitably mimic. Local IC fails, however, because any neighborhood of type $1/2$ contains some types strictly below, all of which would profitably mimic $1/2$. If we did not rule this out in local IC, then the example would contradict \autoref{prop:CC-IC}. \citet[p.~669]{Carroll:12} makes an analogous point in his framework.} When $\Theta\subset \Reals^n$, it is less demanding than imposing only requirement (i) but requiring that all neighborhoods $N_\theta$ in \eqref{e:lIC}
have ``size'' bounded away from zero.\footnote{That is, there is $\epsilon>0$ such that for all $\theta\in \Theta$, the $\epsilon$-ball $\{\theta'\in \Theta: \lVert \theta-\theta'\rVert<\epsilon\}$ is contained in $N_{\theta}$, where $\lVert \cdot \rVert$ is the Euclidean norm.}

Plainly, IC implies local IC (because $\Theta$ is a neighborhood of every type). The converse is not true:
\begin{example}
\label{eg:CCnecessary}
$\Theta=[0,1]$, $A=\{a',a''\}$;
types in $(1/3,2/3)$ strictly prefer $a''$ while all other types strictly prefer $a'$. Note that convex choice fails. The mechanism $m(\theta)=a'$ if $\theta \leq  1/3$ and $m(\theta)=a''$ if $\theta> 1/3$ is not IC. But it is locally IC: for every $\theta\neq 1/3$, there is a neighborhood of $\theta$ on which $m(\cdot)$ is constant; whereas for type $1/3$, neither does it prefer to mimic any type nor does any nearby type prefer to mimic it.
\end{example}

We will see that, in a suitable sense, convex choice is necessary for local IC to imply IC. While it is also sufficient when preferences are strict, the next example shows that indifferences are a threat. In the example, preferences are continuous, the type space is one-dimensional, and the strict-preference regions are convex.

\begin{example}
\label{eg:starnecessary}
$\Theta=[0,1]$, $A=\{a',a''\}$;
types in $[0,\tfrac12)$ strictly prefer $a'$, types in $(\tfrac12,1)$ strictly prefer $a''$, and types $\tfrac12$ and $1$ are indifferent. The mechanism $m(1)=a'$ and $m(\theta)=a''$ for $\theta<1$ is locally IC: type $1$ is indifferent between the two actions, and every other type lies in the open set $[0,1)$ on which $m(\cdot)$ is constant. But it is not IC: every type below $\tfrac12$ strictly prefers reporting $1$.
\end{example}

The problem in \autoref{eg:starnecessary} is the indifference at $\theta=1$, which cannot be approached by types that strictly prefer $a'$. The next definition rules out such
indifferences.

\begin{definition}
\label{def:regularindiffs}
$u$ has \emph{segment-thin indifferences}
if for all $a',a'' \in A$ and $\theta',\theta''\in \Theta$, $u(a',\theta')>u(a'',\theta') \text{ and } u(a',\theta'')= u(a'',\theta'')$ imply there is a sequence $\{\theta_n\}\rightarrow \theta''$ such that $\theta_n\in\ell(\theta',\theta'')$ and $u(a',\theta_n)>u(a'',\theta_n)$ for all $n$.
\end{definition}

An implication of segment-thin indifferences is that for any two actions, if some type strictly prefers one of the actions then the set of types that are indifferent has empty interior. Given convex choice, segment-thin indifferences is a weak requirement;
in particular, it trivially holds when there are no indifferences. 
The conjunction of convex choice and segment-thin indifferences is equivalent to: for all $\theta, \theta'\in \Theta$ and all $a, a' \in A$, if $u(a,\theta)>u(a',\theta)$ and $u(a,\theta')\geq u(a',\theta')$, then $u(a,\lambda \theta+(1-\lambda)\theta')>u(a',\lambda \theta+(1-\lambda)\theta')$ for all $\lambda \in (0,1)$.

We can now state our first main result.

\begin{proposition}
\label{prop:CC-IC}
The following are equivalent:
\begin{enumerate}
\item \label{CC:1} $u$ has convex choice and segment-thin indifferences;
\item \label{CC:2}
for any line segment $\ell\subset \Theta $ and any mechanism $m:\ell \to A$, if $m$ is locally IC
then it is IC.
\end{enumerate}
\end{proposition}

(All proofs are in the appendices.)

The proposition says that convex choice and segment-thin indifferences jointly characterize when 
IC between any two types $\theta$ and $\theta'$ can be determined by just checking local IC along the line segment $\ell(\theta,\theta')$. Such ``integration up'' is a common strategy used to verify IC.

Examples \ref{eg:CCnecessary} and \ref{eg:starnecessary} are illustrative of the general argument for why both convex choice and segment-thin indifferences are necessary for \autoref{prop:CC-IC}'s second statement. Here is a heuristic sketch for why convex choice is sufficient when there are no indifferences. Take any $\theta$ and $\theta'$. Consider a fine grid of types $\theta=\theta_1,\ldots,\theta_K=\theta'$ traversing the line segment $\ell(\theta,\theta')$. We can
regard local IC as implying $u(m(\theta_1),\theta_1)\geq u(m(\theta_2),\theta_1)$, $u(m(\theta_2),\theta_2)\geq u(m(\theta_3),\theta_2)$, and $u(m(\theta_3),\theta_3)\geq u(m(\theta_2),\theta_3)$. Absent indifferences, all those inequalities are strict when the compared allocations differ, as we suppose for this sketch.
Convex choice then implies $u(m(\theta_1),\theta_1)>u(m(\theta_3),\theta_1)$; otherwise, both $\theta_1$ and $\theta_3$ strictly (absent indifferences) prefer $m(\theta_3)$ to $m(\theta_2)$, hence so must $\theta_2$, a contradiction. Iterating this logic, using the combination of local IC and convex choice each time, yields $u(m(\theta_1),\theta_1)>u(m(\theta_i),\theta_1)$ for all $i=2,\ldots,K$. Consequently,  $u(m(\theta),\theta)>u(m(\theta'),\theta)$.

Local IC on the full type space implies local IC on every line segment.\footnote{Formally, if $m:\Theta \to A$ is locally IC, then for every line segment $\ell\subset \Theta$, the mechanism $m':\ell \to A$ defined by $m'(\theta):=m(\theta)$ is locally IC (on $\ell$).} Hence, \autoref{prop:CC-IC} implies that convex choice, together with segment-thin indifferences, ensures sufficiency of local IC on the full type space. But sufficiency on the full type space does not require convex choice, even assuming segment-thin (or no) indifferences. What matters instead is that from any finite choice set, the set of types that find an action uniquely optimal is connected.\footnote{Recall that a set is connected if it cannot be partitioned into two nonempty separated sets, i.e., sets neither of which intersects the closure of the other.}

\begin{definition}
\label{def:connected}
$u$ has \emph{connected choice} if for all finite $B\subset A$ and $a\in B$, the set $\sopt_a^B$ is connected.
\end{definition}

As convex sets are connected, connected choice is weaker than convex choice; indeed, convex choice is connected choice imposed on every line segment, and the notions are equivalent in one dimension. 

Our characterization on the full type space also replaces segment-thin indifferences, which is formulated for line segments, with a sister condition. Let $\cl_{\Theta}$ denote the closure in $\Theta$.

\begin{definition}
\label{def:thin}
$u$ has \emph{thin indifferences} if for all finite $B\subset A$ and $a\in B$, 
\[
\sopt_a^B\neq\emptyset
\quad\Longrightarrow\quad
\wopt_a^B\subset\cl_{\Theta}\sopt_a^B.
\]
\end{definition}

In words, thin indifferences says that from any finite choice set in which some action is uniquely optimal for some type, every type for which that action is optimal is a limit of types for which it is uniquely optimal.

\begin{proposition}
\label{prop:connected-IC}
Assume that for all $a,a'\in A$, the set $\{\theta\in\Theta: u(a,\theta)>u(a',\theta)\}$ is open.
The following are equivalent:
\begin{enumerate}
\item \label{cIC:1} $u$ has connected choice and thin indifferences;
\item \label{cIC:2} every locally IC mechanism with finite range is IC.
\end{enumerate}
\end{proposition}

\autoref{prop:connected-IC} holds verbatim, with the same proof, on any connected topological type space. The intuition for sufficiency of connected choice is as follows. Take a locally IC mechanism $m$ with finite range, fix an action $a$ in that range, and consider the types $\sopt_a^{m(\Theta)}$ for which $a$ is uniquely optimal in the range. By local IC, these types must receive $a$ if they are near a type that receives $a$. As connected choice makes $\sopt_a^{m(\Theta)}$ connected, either all these types receive $a$ or none does. In other words, a type cannot be denied its uniquely optimal action in the mechanism's range if that action goes to another type that also finds it uniquely optimal. The formal proof extends this argument to show there cannot be any violation of IC, using thin indifferences to deal with types whose optimal action is not unique.
In the other direction, a violation of either connected choice or thin indifferences yields a locally IC mechanism that is not IC, in the spirit of Examples \ref{eg:CCnecessary} and \ref{eg:starnecessary}.

The openness assumption in \autoref{prop:connected-IC} holds, in particular, when $u(a,\cdot)$ is continuous for every $a$; unlike continuity, the assumption is ordinal, being preserved under any utility representation. The assumption plays no role in the necessity of connected choice or thin indifferences;\footnote{When $\Theta$ is an arbitrary topological space, this holds provided that all singleton type sets are closed.} it can also be dropped for sufficiency of the conditions if thin indifferences is strengthened
as indicated in the \hyperref[proof:connected-IC]{proposition's proof}.

\begin{remark}
\label{rem:TI-RI}
\hyperref[def:regularindiffs]{Segment-thin indifferences} is \hyperref[def:thin]{thin indifferences} for all binary choice sets, imposed on every line segment; given convex choice, the restriction to binary choice sets is without loss.
In general, neither segment-thin indifferences nor thin indifferences implies the other. However, under the openness assumption of \autoref{prop:connected-IC}, segment-thin indifferences implies thin indifferences.
These claims are proved in \appendixref{proof:TI-RI}.
\end{remark}

\autoref{prop:connected-IC} notwithstanding, we view the sufficiency of local IC on the full type space as generally inadequate to render multidimensional incentive constraints tractable.  In particular, it seems intractable to verify IC between two types by checking local IC along (the typically uncountable set of) all paths between those types. Instead, tractability does obtain if IC can be verified via local IC along their line segment; that is a one-dimensional task. Hence our interest in the more demanding form of sufficiency of local IC in \autoref{prop:CC-IC}. Focusing on line segments also captures a robustness desideratum: a tractable environment ought to remain tractable in its sub-environments of the same kind---here, convex subsets of the type space---of which the compact one-dimensional ones are precisely the line segments.\footnote{Consider an environment with a one-dimensional type space, $\Theta'$, in which local IC does not imply IC; let us stipulate that this environment is not tractable. It would be perverse to then say that embedding it into a multidimensional environment with type space $\Theta\supset\Theta'$ could result in the multidimensional environment being tractable if, on $\Theta$, local IC implies IC. \autoref{fig:discs} below gives an example of such an embedding.}

\autoref{fig:discs} illustrates the difference between sufficiency of local IC on the full type space and on all line segments. The figure's example violates convex choice, but on the full type space any locally IC mechanism is IC because there is connected choice and thin indifferences.
 However, as shown in the figure, there are mechanisms on $\ell(\theta_1,\theta_2)$ that are locally IC even though $\theta_2$ would mimic $\theta_1$.\footnote{In the figure's example,
\begin{equation}
\label{e:Nima}
\text{if $m:\Theta \rightarrow A$ is locally IC on every line segment, then it is IC (on $\Theta$)}.
\end{equation}
So statement \ref{CC:2} of \autoref{prop:CC-IC} is more demanding than \eqref{e:Nima}. The gap is because \eqref{e:Nima} fixes a single mechanism on $\Theta$ and requires local IC of its restriction to every line segment, whereas the proposition's statement allows the mechanism to vary with the line segment.
Under \autoref{prop:connected-IC}'s openness assumption, \eqref{e:Nima} holds for mechanisms with finite range if and only if $u$ has connected choice and satisfies a stronger form of thin indifferences which requires that whenever
$\sopt_a^B\neq\emptyset$, every type in $\wopt_a^B$ is the limit of a sequence in $\sopt_a^B$ contained in a line segment. This can be shown by adapting the proof of \autoref{prop:connected-IC}.}

\begin{figure}
\centering
\begin{tikzpicture}[scale=2]
  \fill[red!20!white] (0,0) circle (1);
  \fill[blue] (0,0) circle (0.5);
    \clip (0,0) circle (1);
    \fill[blue!20!white, draw=red] (0,0) circle (0.5);
   \coordinate (endpoint) at (-.40,-.3);
   \node[label={[above, label distance=0pt, xshift=1.5mm]: \small $\theta'$},circle,fill=black,scale=0.25] (tprime) at (endpoint) {} ;
   \node[label={[above, label distance=0pt]:\small $\theta_1$},circle,fill=black,scale=0.25] (t1) at (-.75,-.3) {} ;
   \node[label={[above, label distance=0pt]: \small $\theta_2$},circle,fill=black,scale=0.25] (t2) at (.75,-.3) {} ;
  \draw (endpoint) -- (t2) ;
  \draw[dotted,thick] (t1) -- (endpoint);
  \draw [decorate,decoration={calligraphic brace,amplitude=5pt,mirror,raise=2pt},yshift=10pt]
  (endpoint) -- (t2) node[midway,below=5pt,black,font=\scriptsize] {$a''$};
  \draw [decorate,decoration={calligraphic brace,amplitude=5pt,mirror,raise=2pt},yshift=10pt]
  (t1) -- (endpoint) node[midway,below=5pt,black,font=\scriptsize] {$a'$};
    \node[red,scale=1.05] at (0.0,0.7) {$a' \succ a''$};
    \node[blue,scale=1.05] at (0.0,0.2) {$a'' \succ a'$};
\end{tikzpicture}
\caption{$\Theta$ is the large disc and $A=\{a',a''\}$. Blue (inner disc) types prefer $a''$ to $a'$; red (outer region) types prefer $a'$ to $a''$; types on the circle between the two regions are indifferent. On the full type space, local IC implies IC. However, when restricted to the line segment from $\theta_1$ to $\theta_2$, a mechanism that allocates $a'$ on the dotted segment (including to $\theta'$) and $a''$ on the solid one is locally IC but not IC.}
\label{fig:discs}
\end{figure}

Let us now compare our approach with that of \citet{Carroll:12} in his ``cardinal type space'' analysis.
We view the approaches as complementary: \citeauthor{Carroll:12} equates types with preferences (more precisely, utilities) and studies properties of the type space that guarantee sufficiency of local IC; instead, we fix an abstract (convex Euclidean) type space and study properties of the utility function. As previously noted, types may have external meaning beyond the agent's preferences. But even through the narrow lens of preferences, \citet[][p.~369]{BMN:17} provide an insightful discussion of ``parameter representations'', like ours, versus ``domain representations'', like \citeauthorpos{Carroll:12}.\label{domainrep}
We add that the flexibility in choosing a parameter representation (the type space and the utility function) has advantages and disadvantages, roughly corresponding to whether one relies on the necessity or the sufficiency of convex choice in \autoref{prop:CC-IC}. We think the latter is more important, and hence the flexibility is more of an advantage.

Mathematically, \autoref{prop:CC-IC} subsumes 
\citepos{Carroll:12} Proposition 1. To see why, note that \citeauthor{Carroll:12}
assumes each $a\in A \subset \Reals^n$
is a probability vector over a finite set of $n$ outcomes, and he associates types with their utility vectors over those outcomes.
We nest his framework by setting $u(a,\theta)=
a \cdot \theta$. Our notion of local IC (and IC) is then equivalent to his. 
\citeauthorpos{Carroll:12} Proposition 1 says that given convexity of $\Theta$, local IC implies IC. The conclusion follows from our \autoref{prop:CC-IC} because $u$ in this specification has convex choice and segment-thin indifferences.

Our result is more general than \citepos{Carroll:12}. First, we do not require finitely many outcomes. Second, and more substantively, there are utility functions that satisfy the conditions of \autoref{prop:CC-IC} but preclude obtaining a convex type space by identifying types with utility vectors over outcomes (so \citeauthorpos{Carroll:12} result has no bite).
For example, let $X=\{0,1/2,1\}$ be a set of three outcomes, $A$ be the lotteries over $X$, and for $\Theta=[0,1]$ let $u(a,\theta)=\sum_{x\in X}a_{x}\left(-|x-\theta|\right)$, i.e., expected utility with a linear-loss Bernoulli utility. As verified in \autoref{eg:Carroll-linear-loss}, this specification has convex choice and segment-thin indifferences, yet because preferences over lotteries pin down each type's utility vector over $X$ up to a positive affine transformation, no choice of these utility vectors yields a convex set.

So far we have emphasized the sufficiency of local IC. But convex choice is relevant and useful more broadly. We next discuss two other applications.

\subsection{Implementability} In mechanism design with transfers, an important question is whether a given allocation rule is implementable. Formally, let $A\equiv Y \times \Reals$, where $Y$ is a
``physical'' 
allocation space and $\Reals$ is the space of transfers.
Writing $a\equiv (y,t)$, assume $u((y,t),\theta)=\tilde{u}(y,\theta)-t$ for some function $\tilde u$. An \emph{allocation rule} $\upsilon:\Theta \to Y$ is \emph{implementable} if there exists a transfer rule $\tau: \Theta \to \Reals$ such that $m\equiv (\upsilon,\tau)$ is IC. We restrict attention to allocation rules with finite range. \citet{rochet87} establishes that an allocation rule is implementable if and only if it is ``cyclically monotone''. A less demanding condition, which only considers pairs of types, is that of \emph{weak monotonicity}:
\begin{equation*}
\forall \theta,\theta' \in \Theta: \; \tilde u(\upsilon(\theta),\theta) -\tilde u(\upsilon(\theta'),\theta) \ge \tilde u(\upsilon(\theta),\theta') - \tilde u(\upsilon(\theta'),\theta').
\end{equation*}
It is straightforward that weak monotonicity is necessary for implementability.
For $Y\subset \Reals^n$ with dot product valuations, $\tilde u(y,\theta)=y\cdot \theta $, \citet{SaksYu05} establish that if $\Theta$ is convex, then weak monotonicity is also sufficient for implementability; see \citet{Bikhchandani06} and \citet{AK14} as well.  Building on \citet[Theorem 5]{BMN:17}, convex choice yields the following generalization, which additionally does not require $Y\subset \Reals^n$.
\begin{proposition}
    \label{prop:weakmon}
    Assume $u((y,t),\theta)=\tilde{u}(y,\theta)-t$ has convex choice and is continuous in $\theta$. If an allocation rule with finite range is weakly monotone, then it is implementable.
\end{proposition}

To interpret the result, note that in this setting with transfers, convex choice holds if and only if for all $y,y'\in Y$, the difference $\tilde u(y,\cdot)-\tilde u(y',\cdot)$ is monotonic along every line segment. For, 
the types that strictly prefer $(y,t)$ to $(y',t')$ are those with $\tilde u(y,\theta)-\tilde u(y',\theta)>t-t'$, and $t-t'$ is arbitrary. \autoref{prop:weakmon} thus extends to multiple dimensions the familiar one-dimensional result that if $\tilde u$ has increasing differences for a suitable ordering of $Y$, then weak monotonicity implies implementability.

\subsection{Cheap Talk}
    \hypertarget{para:cheaptalk}{}%
    In cheap-talk or costly-signaling applications, it is natural to ask whether, in equilibrium, the set of types that choose a particular signal is convex. Modulo details about indifferences, the key condition to guarantee that all equilibria have this ``convex partitional'' property is that the sender's utility have convex choice.

    Indeed, there has been interest in extending \citepos{CS82} classic result on interval equilibria in one-dimensional cheap talk to multiple dimensions. In a working paper version of \citet{LR:07}, \citet{LR:04} assume $\Theta=A= \Reals^2$ and sender utility $u(a,\theta)=-
    \left(\sum_{i \in\{1, 2\}} \alpha_{(i)}\left|a_{(i)}-\left(b_{(i)}+\theta_{(i)}\right)\right|^p\right)^{\frac{1}{p}}$ with parameters $\alpha\in \Reals^2_{>0}$, $b\in \Reals^2$, and $p>1$. It can be verified that this specification has convex choice if and only if $p=2$
    (cf.~\autoref{rem:norms} in the next section). Consistent with that, \citet{LR:04} do not show that all equilibria are convex partitional; rather, in their Section 4.1, they identify parameters under which they can construct some such equilibria.\footnote{In other working paper versions, the authors assume $p=2$ and show that in this case all equilibria are convex partitional.} Similarly, the class of sender utilities in \citet{CH:07} does not assure convex choice, and they only show existence of ``comparative cheap-talk'' equilibria that are convex partitional.

    Also germane are a few papers on common-interest cheap-talk games with an infinite type space but a finite message space.
    \citet{JMR:11} posit $u(a,\theta)=-l\left(\Vert a-\theta \Vert\right)$ for an arbitrary norm $\Vert \cdot \Vert$ and strictly increasing function $l(\cdot)$, but in their Corollary 1 on optimal equilibria being convex partitional (modulo indifferences---which we ignore in the rest of this paragraph), they assume the Euclidean norm, which guarantees convex choice. They note that other norms can be problematic; we will see that weighted Euclidean norms would also work, but none others (\autoref{rem:norms}).\footnote{Strictly speaking, \citepos{JMR:11} Example 1/Figure 1 with the maximum norm does not violate convex choice---rather, the issue there is about indifferences---but perturbations of their example's actions would illustrate a failure of convex choice.} \citet[][Theorem 2]{SP:17} establishes that optimal equilibria are convex partitional under a weaker condition than \citet{JMR:11}. \citeauthor{SP:17}'s condition implies the directional single-crossing property discussed in the next section, which we show essentially characterizes convex choice (\autoref{prop:CC-DSCD}).
    \citet[][Proposition 3]{Sobel:16} extends \citeauthor{SP:17}'s result by directly assuming convex choice.
    Finally, \citet[][Proposition 2]{Rodriguez26} shows that when there are multiple senders with common information, message profiles in efficient equilibria have ``convex meaning''. 

\section{Directional Single Crossing}
\label{sec:DSCD}
To unpack which preferences have convex choice, we next develop a connection with single crossing.

\begin{definition} \label{def:DSC} A function $f:\Theta\rightarrow \R$
is \emph{directionally single crossing (DSC)} if there is $\alpha\in\R^n\setminus\{0\}$ such that for all $\theta,\theta'\in \Theta$, 
\begin{align*}
    (\theta-\theta')\cdot \alpha \ge 0 \implies \sign\left(f(\theta)\right)\ge \sign\left(f(\theta')\right).%
\end{align*} 
\end{definition}

DSC says that for any $\theta$ and $\theta'$ such that $\theta$ is to the right of the hyperplane passing through $\theta'$ in the direction of $\alpha$ (i.e., with normal vector $\alpha$), we have $\sign\left(f(\theta)\right)\geq \sign\left(f(\theta')\right)$. Hence, we will sometimes say more explicitly that a function is DSC in the direction $\alpha$. DSC is equivalent to the sets $\{\theta:f(\theta)<0\}$ and $\{\theta:f(\theta)>0\}$ being either empty, full (i.e., all of $\Theta$), or intersections of $\Theta$ with half-spaces defined by parallel (possibly identical) hyperplanes.
See \autoref{fig:DSC}: the left panel depicts a typical DSC function with a ``thin'' zero set, illustrating the geometry of the definition; the middle panel depicts a DSC function with a ``thick'' zero set; and the right panel depicts a violation of DSC. Note that DSC implies that along any line segment, the sign of the function is monotonic.%

\begin{figure}[h]
    \centering
    \begin{subfigure}{0.3\textwidth}
        \centering
        
    \begin{tikzpicture}
    
    \draw[thick,red] (-1,-2) -- (1.5,1.5) node[above] {$f(\cdot)=0$};
    \path[fill=blue,opacity=0.1] (-1,-2) -- (1.5,1.5) -- (3,1.5) -- (3,-2) -- cycle; 
    \coordinate (theta_prime) at (0.25, -0.25); 
    \fill (theta_prime) circle (2pt) node[left] {$\theta'$};

    \draw[thick,->,blue] (theta_prime) -- ++(1.4, -1) node[above right] {$\alpha$};
    \draw[rotate around={-125:(theta_prime)}] (theta_prime) rectangle ++(0.2,0.2);

    \node at (.2,1.2) {$f(\cdot)<0$};
    \node at (2,0.7) {$f(\cdot)>0$};
    \node[blue] at (1.2,-1.8) {\small $\alpha\cdot (\theta-\theta')\ge 0$};
    \end{tikzpicture}
        \caption{DSC}
        \label{fig:DSC1}
    \end{subfigure}
    \textcolor{gray}{\vrule width 1.5pt}
    \quad
    \begin{subfigure}{0.3\textwidth} 
    \begin{tikzpicture}    
        \coordinate (shifter) at (0.6, -0.6/1.4); 
    
        \fill[red!10!white] (-1, -2) -- (1.5, 1.5) -- ($( 1.5, 1.5 ) + (shifter)$) -- ($( -1, -2 ) + (shifter)$) -- cycle;

        \draw[thick,black] (-1,-2) -- (1.5,1.5);
        \draw[thick,red] ($( -1, -2 ) + (shifter)$) -- ($( 1.5, 1.5 ) + (shifter)$);
    
        \path ($(1.15, -0.25)!0.5!($(0.25, -0.25) + (shifter)$)$) coordinate (labelpos);
        \node[rotate=55, anchor=south, text=red] at (labelpos) {$f(\cdot)=0$};
    
        \node at (-0.2,.75) {$f(\cdot)<0$};
        \node at (1.5,-1.5) {$f(\cdot)>0$};
    \end{tikzpicture}
        \caption{DSC with thick zero set}
        \label{fig:DSC2}
    \end{subfigure}
    \textcolor{gray}{\vrule width 1.5pt}
    \begin{subfigure}{0.3\textwidth}
        \centering
        \begin{tikzpicture}            
            \draw[thick,red]  (-1,-2) .. controls (-.5,.5)  .. (1.5,1.5) node[right] {$f(\cdot)=0$};
        
            \node at (-.5,1.3) {$f(\cdot)<0$};
            \node at (.5,-.5) {$f(\cdot)>0$};
        \end{tikzpicture}
        \caption{Not DSC}
        \label{fig:DSC3}
    \end{subfigure}
    \caption{Directional single crossing with $\Theta \subset \Reals^2$.}
    \label{fig:DSC}
\end{figure}

\begin{definition}
\begin{samepage}
The utility function $u$ has
\emph{directionally single-crossing differences
(DSCD)} if for all $a,a'\in A$, the difference $D_{a,a'}(\theta):=u(a,\theta)-u(a',\theta)$ is
DSC.%
\end{samepage}
\end{definition}

DSCD can be viewed as saying that for any pair of actions $a$ and $a'$, the strict-preference sets $\{\theta:u(a,\theta)>u(a',\theta)\}$ and $\{\theta:u(a,\theta)<u(a',\theta)\}$ are parallel half-spaces, each either open or closed, intersected with the type space. 
Importantly, the direction of the hyperplanes defining these half-spaces can vary across action pairs. 

Here are two leading families of DSCD when $A \subset \Reals^n$: (i) weighted Euclidean preferences, where  $u(a,\theta)$ is any strictly decreasing function of $(a-\theta) W (a-\theta)^{T}$ with $W$ an $n\times n$ symmetric positive definite matrix; and (ii) constant-elasticity-of-substitution (CES) preferences, where $A,\Theta\subset \Reals^n_{\geq 0}$ and $u(a,\theta)=\left(\sum_{i=1}^n (a_{(i)})^r\theta_{(i)} \right)^{s}$ with parameters $r>0$ and $s>0$.\footnote{The superscript $^T$ denotes transposition. Family (i) satisfies DSCD with the direction $\alpha = W (a-a')^{T}$ and family (ii) with $\alpha = a^r-(a')^r$. Both families also have segment-thin indifferences.} 
Adding an arbitrary type-independent function $w(a)$ preserves DSCD in family (i) when $u(a,\theta)$ is linear in $(a-\theta) W (a-\theta)^{T}$ and in family (ii) when $s=1$, because in both cases pairwise utility differences are affine in type.
The latter includes $u(a,\theta)=a\cdot\theta+w(a)$, which is frequently used in multidimensional mechanism design and which we will return to.

\begin{remark}\label{rem:norms}Consider norm-based preferences, where $A\subset\Reals^n$
and $u(a,\theta)=-l(\lVert a-\theta\rVert)$ for a strictly increasing function $l(\cdot)$ and some norm $\lVert\cdot \rVert$ on $\Reals^n$.
If the norm is weighted Euclidean (i.e., there is a symmetric positive definite $n\times n$ matrix $W$ such that $\lVert x \rVert = \sqrt{xWx^T}$ for all $x\in \Reals^n$), then we have weighted Euclidean preferences and hence DSCD. Conversely, if the norm is not weighted Euclidean, then DSCD is violated if $A\cap \Theta$ has nonempty interior; \appendixref{proof:norms} provides a proof.
\end{remark}

The following result shows that DSCD ``almost'' characterizes convex choice.  Loosely speaking, $u$ \emph{strictly violates DSCD} if for some pair of actions, the strict preference between them switches more than once in every direction; the formal definition is deferred to
\appendixref{proof:CC-DSCD}.

 \begin{proposition}
 \label{prop:CC-DSCD}
 If $u$ has DSCD, then $u$ has convex choice.
 If $u$ has convex choice, then it does not strictly violate DSCD.
 \end{proposition}

In light of its half-spaces interpretation, DSCD implying convex choice is straightforward. Conversely, a separating hyperplane theorem yields that under convex choice there is no strict violation of DSCD
(although DSCD itself can fail; see \autoref{eg:strictviolation}).

\autoref{prop:CC-DSCD} is closely related to a result of \citet[p.~322]{Grandmont:78}. \citeauthor{Grandmont:78}'s  condition (H.2) is a little stronger than imposing both convex choice and segment-thin indifferences. Combined with a continuity condition---\citeauthor{Grandmont:78}'s (H.1)---and an assumption that the convex type space is open, he obtains a characterization that is similar to, but more restrictive than, DSCD. See \appendixref{sec:Grandmont} for details.
The leading families of DSCD given above satisfy \citeauthor{Grandmont:78}'s characterization. A one-dimensional example of DSCD that does not is $\Theta=(0,1)$, $A=\{a',a''\}$, $u(a',\cdot)=0$, $u(a'',\theta)=-1$ if $\theta<1/2$ and $u(a'',\theta)=1$ if $\theta\geq 1/2$; \citeauthor{Grandmont:78}'s continuity requirement fails here.

If $\Theta \subset \Reals$, then DSCD is equivalent to \citepos{KLR:23} ``single-crossing differences'' and, modulo indifferences, \autoref{prop:CC-DSCD} is equivalent to part 1 of their Theorem 1. But in multiple dimensions the notions are distinct, and their result is not about convex choice. Similarly, \citepos{MS:94} single-crossing preferences do not guarantee convex choice when $\Theta$ is multidimensional.\footnote{Neither does \citepos{BS:18} ``$i$-single crossing property''. Their notion is motivated by \citepos{Quah:07} weakening of \citepos{MS:94} comparisons of constraint/choice sets. Convex choice considers all choice sets.}

\citet{MM:88} present a notion of ``generalized single crossing''. Their condition is formulated for a Euclidean allocation space and twice-differentiable utilities. For differentiable mechanisms, they establish that their condition implies that checking local IC along line segments is sufficient. \appendixref{sec:MM} provides an example demonstrating that \citeauthor{MM:88}'s condition does not imply convex choice, and hence does not imply DSCD. Consequently, as displayed explicitly in our example, there are non-differentiable mechanisms for which their condition does not yield sufficiency of local IC.

\subsection{Convex Environments}
\label{sec:convexenv}

In various settings, the agent may be choosing among lotteries. For instance, a mechanism designer may consider stochastic mechanisms; indeed, the revelation principle in general requires allowing for those. Or, the mechanism may take inputs from other agents with private information, so that from the interim point of view of one agent, her report induces a lottery. Alternatively, in a cheap-talk application, the receiver's decision may be determined by both the sender's message and the receiver's private preference type; \citet{KLR:23} detail this application with $\Theta,A \subset \Reals$.

Following \citet{KLR:23}, we say that the environment $(A,\Theta,u)$ is \emph{convex} if
\begin{equation} 
\tag{$\star$}
\text{ the set of functions } \{u(a,\cdot):\Theta\to \Reals\}_{a \in A} \text{ is convex.}
\label{star}
\end{equation}
That is, a convex environment has the property that for any pair of actions and any weighting, there is a third action that replicates the weighted sum of utilities of the original actions. Note that convexity is a property of the utility function rather than of $A$. However, if $A$ is convex, then it is straightforward that \eqref{star} is assured by linearity of $u$ in its first argument. Expected utility over a convex set of lotteries (e.g., all lotteries) is thus a convex environment. In fact, rank-dependent expected utility also produces a convex environment \citep[Example 2]{KLR:23}. An example without lotteries is $A=[\underline a_1,\overline a_1]\times \ldots \times [\underline a_k,\overline a_k]\subset \Reals^k$ and $u(a, \theta)= \sum_{i=1}^k g_i\left(a_{(i)}\right)f_i(\theta)$ with arbitrary functions $\left(f_i\right)_{i=1}^k$ and continuous functions $\left(g_i\right)_{i=1}^k$.

The following concepts are useful to elucidate the implications of DSCD in convex environments. We say that the utility function $\tilde u(a,\theta)$ \emph{represents} $u(a,\theta)$ if $\tilde{u}(a,\theta)=h(u(a,\theta),\theta)$ for some function $h$ such that for all $\theta$, $h(\cdot,\theta):\Reals\to \Reals$ is strictly increasing. This is simply the standard notion of (type-dependent) preference representation. An \emph{affine representation} is any representation that is affine in $\theta$, i.e., has the form $v(a)\cdot \theta+w(a)$ for some $v:A\to \Reals^n$ and $w:A\to \Reals$. We say that $u$ is \emph{one-dimensional} if there are $\alpha \in \Reals^n\setminus\{0\}$ and $\tilde u: A\times\Reals \rightarrow \Reals$ such that $u(a,\theta)\ge (>) u(a',\theta)$ if and only if $\tilde u(a,\alpha\cdot \theta)\ge (>) \tilde u(a',\alpha\cdot \theta)$.\footnote{For a DSCD preference, this is equivalent to the utility difference between every pair of actions satisfying DSC in a direction parallel to $\alpha$.} In other words, any type $\theta$'s preferences are fully determined by the one-dimensional sufficient statistic $\alpha \cdot \theta$. A type $\theta$ is \emph{totally indifferent} if $u(a,\theta)=u(a',\theta)$ for all $a,a'\in A$.

It bears emphasis that our notion of representation does not permit redefining types, only the utility function. If types could be redefined, including altering the type space, then expected utility over lotteries with a finite number of outcomes would always have an affine representation, simply by viewing each type as the vector of its utilities over outcomes.
This is not true once types are fixed
(\autoref{eg:one-dim-DSCD-no-affine-rep}).

Regardless of a convex environment, any utility function with an affine representation satisfies DSCD. For, any affine utility clearly satisfies DSCD, and DSCD is an ordinal property that is preserved by any representation. If preferences are one-dimensional, then even in a convex environment there are DSCD utilities that do not have an affine representation (\autoref{eg:one-dim-DSCD-no-affine-rep}). The following result says that if preferences are not one-dimensional, then under some ``regularity'' conditions, DSCD in a convex environment is in fact characterized by an affine utility representation.

\begin{proposition}
\label{prop:DSCDstar}
    Assume $\Theta=\R^n$, $u(a,\theta)$ is differentiable in $\theta$, and no type is totally indifferent.\footnote{\label{fn:totalindiff}As explained in the proof, instead of no totally indifferent type, it is enough to assume that for any type $\theta$ that is totally indifferent, there are two actions $a'$ and $a''$ such that $\nabla[u(a',\theta)-u(a'',\theta)]\neq 0$.} If the environment is convex and $u$ is not one-dimensional, then $u$ has DSCD if and only if it has an affine representation.
\end{proposition}

\paragraph{Proof idea.} We pick an arbitrary action $a_*$ and normalize $u(a_*,\cdot)=0$. DSCD and the convex environment then imply that every linear combination of the functions $u(a,\cdot)$, $a\in A$, is DSC.

Suppose that preferences are not one-dimensional. Then there are two such combinations, $u_1$ and $u_2$, whose DSC directions cannot be chosen to be parallel, so their zero sets contain non-parallel hyperplanes. Since $\Theta=\R^n$, these hyperplanes intersect, and so $u_1$ and $u_2$ vanish at some common type $\theta_0$. For simplicity, suppose also that $\nabla u_1(\theta_0)\neq0$; the formal proof shows how to arrange this, using that no type is totally indifferent.

The key observation is that a DSC function that vanishes at two types vanishes on the line segment between them. Fix any $\theta_1\in\Theta$. The combination $u_2(\theta_1)u_1-u_1(\theta_1)u_2$ is DSC and vanishes at $\theta_0$ and $\theta_1$, so its derivative at $\theta_0$ in the direction $\theta_1-\theta_0$ is zero: writing $L_i(\theta):=\nabla u_i(\theta_0)\cdot(\theta-\theta_0)$ for $i=1,2$,
\[ u_2(\theta_1)L_1(\theta_1)=u_1(\theta_1)L_2(\theta_1). \]
We now multiply all utilities by $L_1(\theta)/u_1(\theta)$. DSC and $\nabla u_1(\theta_0)\neq0$ ensure that this factor is positive wherever it is defined, so the rescaled utilities still represent $u$. The rescaled $u_1$ is $L_1$ and, by the preceding display, the rescaled $u_2$ is $L_2$; both are affine.

Finally, choose a type where $u_2$ vanishes but $u_1$ does not. From any other utility $u(a,\cdot)$, subtract a multiple of $u_1$ so that the difference also vanishes at this type. Repeating the argument above at this type, with $u_2$ in place of $u_1$, we can show that the rescaled difference is affine, which implies that the rescaled $u(a,\cdot)$ is affine. The formal proof verifies that the rescaling can also be defined on the zero set of $u_1$.
\marker

Recall that DSCD essentially characterizes convex choice (\autoref{prop:CC-DSCD}), and we have argued that convex choice is crucial for tractability (\autoref{prop:CC-IC}) and appealing for other reasons (e.g., in \hyperlink{para:cheaptalk}{cheap-talk games}).
Moreover, in many problems the set of lotteries to consider cannot be substantially restricted; this can arise with stochastic mechanisms, or with multiple agents and deterministic mechanisms, where each agent's Bayesian incentive constraints involve allocation lotteries induced by the other agents' types. We thus view \autoref{prop:DSCDstar} as suggesting that in many settings, genuinely multidimensional expected-utility preferences without an affine representation will be unwieldy.
This may shed light, in particular, on the prominence of the affine form in multidimensional mechanism design.\footnote{Some work with a multidimensional type space and non-affine utilities effectively assumes preferences are one-dimensional \citep[e.g.,][]{Ghili:23}.} It is generally difficult in those problems to rule out optimality of stochastic mechanisms, even under the affine form; see, e.g., \citet{Pycia06} and \citet{MV:07}.
By contrast, there are canonical one-dimensional settings in which ``regularity'' assumptions do ensure that deterministic mechanisms are optimal \citep{Strausz06}; in these cases, solutions have been obtained assuming only DSCD over deterministic allocations.\footnote{\citepos{Strausz06} result is for a single-agent setting with transfers and quasi-linear utility. Beyond that, even one-dimensional mechanism design frequently assumes functional forms that assure DSCD over lotteries; \autoref{rem:1d} details the relevant restrictions. Examples in the context of delegation include \citet{AB20} and \citet{KKVW21}.}

\autoref{eg:spiraling} shows that the assumption of $\Theta=\Reals^n$ cannot be dropped from \autoref{prop:DSCDstar}. But one can obtain the result with alternative assumptions; what is important is some richness, beyond the convexity requirement \eqref{star}, in terms of the family $\{u(a,\cdot)\}_{a\in A}$ when preferences are not one-dimensional.  
We offer one such variant below. Say that an action $a\in A$ \emph{strictly dominates} $a'\in A$ if $u(a,\theta)>u(a',\theta)$ for all $\theta\in\Theta$. We also say that
 $u$ is \emph{minimally rich} if there are no actions $a_0,a_1,a_2\in A$ and functions $\lambda_0,\lambda_1: A\rightarrow \R$ such that for all $a\in A$ and $\theta\in\Theta$, $$u(a,\theta)-u(a_0,\theta) = \lambda_0(a) [u(a_1,\theta)-u(a_0,\theta)] +\lambda_1(a)[u(a_2,\theta)-u(a_0,\theta)].$$
 That is, minimal richness fails if, after normalizing $u(a_0,\cdot)=0$,
every action $a$'s utility $u(a,\cdot)$ is a linear combination of $u(a_1,\cdot)$ and $u(a_2,\cdot)$.

\begin{proposition}\label{prop:DSCDstar-strict-dominance}
    Assume $u(a,\theta)$ is differentiable in $\theta$, minimally rich, and there is an action that strictly dominates another. If the environment is convex,
    then $u$ has DSCD if and only if it has an affine representation.    
\end{proposition}

\autoref{prop:DSCDstar-strict-dominance} substitutes \autoref{prop:DSCDstar}'s assumptions of $\Theta=\Reals^n$ and no total indifference with minimal richness and strict dominance between some pair of actions.\footnote{\autoref{prop:DSCDstar}'s assumption that $u$ is not one-dimensional is implied by minimal richness and DSCD in a convex environment.} The dominance assumption is satisfied whenever there is one component of the action space over which all types have strictly monotonic preferences; in particular, it holds in the quasi-linear environments common in mechanism design.

The ``one-dimensional-or-affine'' message of Propositions \ref{prop:DSCDstar} and \ref{prop:DSCDstar-strict-dominance} can be illustrated with our leading DSCD families. First consider $X,\Theta \subset \Reals^n_{\geq 0}$ each with nonempty interior, and the CES utility $\bernoulli(x,\theta)=\left(\sum_{i=1}^n (x_{(i)})^r\theta_{(i)} \right)^{s}$ with $r>0$ and $s>0$. Although $\bernoulli$ satisfies DSCD, does the induced expected-utility function $u(a,\theta)$ over the lottery space $A=\Delta X$? In one dimension, $n=1$, yes: in that case any linear combination $\lambda \bernoulli(x,\theta)+\lambda' \bernoulli(x',\theta)=(\lambda x^{rs}+\lambda' (x')^{rs})\theta^s$
is monotonic in $\theta$, hence (directionally) single crossing. But in multiple dimensions, $n>1$, $u$ has DSCD if and only if $s=1$; indeed, for $s\neq1$ even convex choice fails.%
\footnote{Choose $x_0$ and $\theta_0$ in the interiors of $X$ and $\Theta$, and a vector $d\neq0$ with $d\cdot\theta_0=0$. For sufficiently small $\epsilon,t>0$, the outcomes $x_\pm:=(x_0^r\pm\epsilon d)^{1/r}$ belong to $X$ and the types $\theta_\pm:=\theta_0\pm td$ belong to $\Theta$. Let $a$ be the equiprobable lottery on $x_+$ and $x_-$, and $a'$ the sure outcome $x_0$. Type $\theta_0$ is indifferent, since $x_\pm^r\cdot\theta_0=x_0^r\cdot\theta_0$. At $\theta=\theta_\pm$, the numbers $x_+^r\cdot\theta$ and $x_-^r\cdot\theta$ are distinct and average to $x_0^r\cdot\theta$, so Jensen's inequality implies that both $\theta_+$ and $\theta_-$ strictly prefer $a$ to $a'$ if $s>1$, and $a'$ to $a$ if $s<1$. Since $\theta_0$ is the midpoint of $\theta_+$ and $\theta_-$, convex choice is violated.} As already mentioned, $s=1$ is the affine case prominent in multidimensional mechanism design. 

Furthermore, when $s=1$, adding a function of the outcome alone keeps $\bernoulli$ affine in type. Since adding a function of type alone does not change preferences, the case of $s=1$ thereby subsumes quadratic-loss preferences, $\bernoulli(x,\theta)=-\Vert x-\theta \Vert^2$ with $\Vert \cdot \Vert$ the Euclidean norm. Interestingly, in his study of common-interest cheap talk with noisy communication---so that the sender's message induces a lottery over receiver decisions---\citet[][Proposition 6.3]{Bauch:24} requires quadratic loss to obtain a convex-partitional equilibrium result. Without noise, \citet[][Corollary 1]{JMR:11} could allow for preferences based on more general functions of the Euclidean norm.\footnote{Cheap-talk aficionados may note that \citet{BBK:07} get interval equilibria in their one-dimensional model of noisy communication without imposing quadratic loss on the sender's utility. That owes to their special ``truth-or-noise'' garbling: either the sender's message goes through correctly, or a random message is drawn from an exogenous distribution. So when choosing her message, the sender can condition on the no-noise event, effectively choosing among deterministic receiver actions. In a multidimensional setting, including \citepos{Bauch:24}, truth-or-noise garbling would ensure that all equilibria are convex partitional (modulo details about indifferences) if the sender's utility has DSCD, even if it does not have DSCD over lotteries.} Our results shed light on this contrast: weighted Euclidean preferences have DSCD, but when $n\geq 2$ and $\Theta=\Reals^n$, \autoref{rem:quadratic} shows that \hyperref[rem:norms]{norm-based preferences} have DSCD over lotteries if and only if the loss is quadratic in a weighted Euclidean norm.

\begin{remark}
For convex environments, \citet{KLR:23} provide a characterization of preferences that satisfy their single-crossing differences (SCD). As \autoref{eg:SCD-no-affine-rep} shows, our message that multidimensional preferences must have an affine representation does not obtain under SCD, because SCD corresponds to interval choice whereas DSCD to convex choice.
\end{remark}

\section{Conclusion}\label{sec:conclusion}

Convex choice is a
valuable property.
We have shown that (i)
subject to a mild condition on indifferences,
it characterizes when local IC is sufficient on all line segments (\autoref{prop:CC-IC}), (ii) it simplifies verifying the implementability of allocation rules (\autoref{prop:weakmon}) and is relevant to applications like \hyperlink{para:cheaptalk}{cheap talk}, (iii) it is essentially equivalent to a form of single crossing with a simple geometric interpretation (\autoref{prop:CC-DSCD}), and (iv) in convex environments satisfying some regularity conditions, it reduces to ``one-dimensional or affine representation'' (Propositions \ref{prop:DSCDstar} and \ref{prop:DSCDstar-strict-dominance}). Convex choice has also been used in other contexts, such as preference aggregation \citep{Grandmont:78,CN:88} and social learning \citep{KLLR24}.

An alternative notion---generally weaker and equivalent only in one dimension---is \hyperref[def:connected]{connected choice}: the set of types that find an action (uniquely) optimal is connected. Subject to some regularity conditions, connected choice characterizes when local IC is sufficient on the full type space (for mechanisms with finite range; \autoref{prop:connected-IC}), rather than on every line segment.

We find convex choice more appealing for three related reasons. First, it ensures that IC between any two types can be verified via local IC along the line segment joining them; this is more tractable than checking local IC along all paths connecting the two types.
Indeed, our proof of \autoref{prop:weakmon} relies on convex choice. Second, convex choice sets are more likely to be economically meaningful; for instance, in multidimensional cheap talk, equilibria that merely partition the sender's type space into connected sets would be less satisfactory. Third, convex choice ties to \hyperref[sec:DSCD]{directional single crossing} and its consequences; connected choice has no such analog.

\vspace{.25in}

\begin{singlespace}
    \addcontentsline{toc}{section}{References}
    \bibliographystyle{ecta}
    \bibliography{references_dsc}
\end{singlespace} 

\vspace{.25in}

\appendix
\numberwithin{example}{section}
\numberwithin{remark}{section}
\numberwithin{proposition}{section}
\numberwithin{lemma}{section}
\numberwithin{definition}{section}

\titleformat{\section}{\color{DarkRed}\normalfont\Large\bfseries}{Appendix \thesection:}{.5em}{}

\section{%
Main Appendix}

\subsection{Proof of \autoref{prop:CC-IC}}
 
\begin{proof}[Proof that statement \ref{CC:1} $\implies$ statement \ref{CC:2}.] 
Assume convex choice and segment-thin indifferences. Fix an arbitrary line segment $\Theta'\subset \Theta$ and let $m:\Theta'\rightarrow A$ satisfy local IC. Fix any two distinct types in $\Theta'$; we can label the line segment between these types as $[0,1]$, corresponding to the convex combinations of the two types. Our goal is to show that $u(m(0),0)\geq u(m(1),0)$, which implies IC on $\Theta'$.

Local IC implies that each type $\theta\in[0,1]$ has a neighborhood $N_{\theta}\subset [0,1]$ satisfying condition \eqref{e:lIC}. For each $\theta$, there is an open ball $B_{\theta}\subset N_{\theta}$ centered at $\theta$. These balls form an open cover of $[0,1]$; as $[0,1]$ is compact, there is a finite subcover $\{B_{\theta_i}\}_{i=1}^K$. We assume without loss that this subcover is minimal and its indices satisfy $\theta_1<\ldots<\theta_K$.

Then $0 \in B_{\theta_1}$ and $1 \in B_{\theta_K}$. Local IC implies
\begin{equation}
\label{e:CC:1} u(m(0),0)\ge u(m(\theta_1),0).    
\end{equation}
Now choose any $\tilde{\theta}_1\in B_{\theta_1}\cap B_{\theta_2}$ with $\theta_1<\tilde{\theta}_1<\theta_2$. By local IC, 
\begin{align}
\label{e:CC:3} u(m(\theta_1),\theta_1)\ge u(m(\tilde{\theta}_1),\theta_1),\\
\label{e:CC:4} u(m(\tilde{\theta}_1),\tilde{\theta}_1)\geq u(m(\theta_1),\tilde{\theta}_1).
\end{align}
We claim that 
\begin{equation}
\label{e:CC:5} u(m(\theta_1),0)\ge u(m(\tilde{\theta}_1),0).
\end{equation}
To prove inequality \eqref{e:CC:5}, suppose to the contrary $ u(m(\tilde{\theta}_1),0)>u(m(\theta_1),0)$. Then inequality \eqref{e:CC:3} implies $\theta_1>0$. If inequality \eqref{e:CC:4} holds strictly, then convex choice implies $u(m(\tilde{\theta}_1),\theta_1)>u(m(\theta_1),{\theta}_1)$, contradicting \eqref{e:CC:3}. If \eqref{e:CC:4} holds with equality, then since $\theta_1\in (0,\tilde{\theta}_1)$, it follows from segment-thin indifferences and convex choice that $u(m(\tilde{\theta}_1),\theta_1)>u(m(\theta_1),{\theta}_1)$, contradicting \eqref{e:CC:3}.

Combining \eqref{e:CC:1} and \eqref{e:CC:5}, $u(m(0),0)\ge u(m(\tilde{\theta}_1),0)$. Continuing in the same fashion, we obtain $u(m(0),0)\ge u(m(1),0)$, as desired.
\end{proof}

\begin{proof}[Proof that statement \ref{CC:2} $\implies$ statement \ref{CC:1}] 
We prove the contrapositive, first for convex choice and then for segment-thin indifferences.

If $u$ violates convex choice, then there are $a',a'' \in A$ and $\theta_1,\theta_2\in \Theta$ and $\theta'\in\ell(\theta_1,\theta_2)$ such that $u(a'',\theta_i)>u(a',\theta_i)$ for $i=1,2$ and $u(a'',\theta')\le u(a',\theta')$.
Let $\Theta':=\ell(\theta_1,\theta_2)$ and let $a^*(\theta)$ be an arbitrary selection from $\argmax_{a\in \{a',a''\}} u(a,\theta)$. Consider the mechanism $m:\Theta' \to \{a',a''\}$ given by
$$
m(\theta):=
\begin{cases}
a^*(\theta) & \text{ if } \theta \in \ell(\theta_1,\theta')\setminus\{\theta'\}\\
a' & \text{ otherwise.}
\end{cases}
$$
This mechanism is locally IC because all types in $\ell(\theta_1,\theta')$ are getting an optimal action from $\{a',a''\}$ and the mechanism is constant on $\ell(\theta',\theta_2)$. But IC fails because type $\theta_2$ can profitably mimic $\theta_1$.

If $u$ violates segment-thin indifferences,
then there are $a',a'' \in A$ and $\theta_1,\theta_2\in \Theta$ and $\theta'\in\ell(\theta_1,\theta_2)\setminus\{\theta_1,\theta_2\}$ such that $u(a',\theta_1)=u(a'',\theta_1)$, $u(a'',\theta_2)>u(a',\theta_2)$, and $u(a',\theta')\geq u(a'',\theta')$. 
Consider the same mechanism $m$ defined in the previous paragraph, except that $m(\theta_1):=a''$.
It is locally IC but not IC by the argument given in the previous paragraph.
\end{proof}

\subsection{Proof of \autoref{prop:connected-IC}}
\label{proof:connected-IC}
\begin{proof}[Proof that statement \ref{cIC:1} $\implies$ statement \ref{cIC:2}.]
Fix a locally IC mechanism $m$ with finite range and suppose toward contradiction that $m$ is not IC. First, replace every action in the range by a fixed representative of its payoff-equivalence class (where $a\sim a'$ if $u(a,\cdot)=u(a',\cdot)$). This changes no utility comparison, so it preserves local IC---with the same neighborhoods---and the failure of IC. We may thus assume that the range $R:=m(\Theta)$ contains no two payoff-equivalent actions.

Second, we claim that types with a unique maximizer in $R$ are dense in $\Theta$. For distinct $b,c\in R$, the tie set $T_{bc}:=\{\theta: u(b,\theta)=u(c,\theta)\}$ is closed, being the complement of the union of two open strict-preference sets, and it has empty interior: since $b$ and $c$ are not payoff-equivalent, one of $\sopt_b^{\{b,c\}}$ and $\sopt_c^{\{b,c\}}$ is nonempty, say the former, and thin indifferences yields $T_{bc}\subset \wopt_b^{\{b,c\}}\subset \cl_{\Theta}\sopt_b^{\{b,c\}}$; a nonempty open subset of $T_{bc}$ would then have to intersect $\sopt_b^{\{b,c\}}$, which is disjoint from $T_{bc}$. As a finite union of closed sets with empty interior has empty interior, the set $G:=\Theta\setminus \bigcup_{b\neq c\in R} T_{bc}$ is open and dense, and every type in $G$ strictly ranks all elements of $R$.

Third, we identify a type with strict preferences and an incentive to misreport. Since $m$ is not IC, there is $\hat\theta$ with $u(m(\hat\theta),\hat\theta)<\max_{r\in R}u(r,\hat\theta)$; let $M$ be the set of maximizers of $u(\cdot,\hat\theta)$ in $R$, and fix $b'\in M$. The set $\sopt_{b'}^{\{b'\}\cup(R\setminus M)}$ of types that strictly prefer $b'$ to every element of $R\setminus M$ contains $\hat\theta$ by the definition of $M$, and it is open, being a finite intersection of strict-preference sets. Hence its intersection with $N_{\hat\theta}$ contains some $\theta_0\in G$. Local IC at $\hat\theta$ implies $m(\theta_0)\notin M$; hence $u(b',\theta_0)>u(m(\theta_0),\theta_0)$. Since $\theta_0\in G$, some $b\in R$ is uniquely optimal in $R$ at $\theta_0$, and $u(b,\theta_0)\geq u(b',\theta_0)>u(m(\theta_0),\theta_0)$. To summarize: type $\theta_0$ strictly prefers $b$ to every other alternative in $R$, i.e., $\theta_0\in\sopt_b^{R}$, but $m(\theta_0)\neq b$.

We will derive a contradiction by enlarging a choice set on which some type receives $b$ and finds it uniquely optimal. Let $A'\subset R$ be a maximal set containing $b$ such that $b\in m\left(\sopt_b^{A'}\right)$. Such a set exists because $A'=\{b\}$ satisfies this condition and $R$ is finite. Note that $\theta_0\in\sopt_b^{R}\subset\sopt_b^{A'}$. By connected choice, $\sopt_b^{A'}$ is connected. It contains a type assigned $b$ and a type assigned another action, namely $\theta_0$. Hence, there are $\theta_1,\theta'\in\sopt_b^{A'}$ with $m(\theta_1)=b\neq m(\theta')$ such that $\theta'\in N_{\theta_1}$ or $\theta_1\in N_{\theta'}$. Otherwise, $\{\theta\in\sopt_b^{A'}:m(\theta)=b\}$ and $\{\theta\in\sopt_b^{A'}:m(\theta)\neq b\}$ would each contain $N_\theta\cap\sopt_b^{A'}$ for each of their members $\theta$, and thus partition $\sopt_b^{A'}$ into two nonempty, relatively open sets. Let $c:=m(\theta')$. Local IC implies
$
u(b,\theta_1)\geq u(c,\theta_1)
$ and $
u(c,\theta')\geq u(b,\theta').
$
Since $\theta'\in\sopt_b^{A'}$, the second inequality implies $c\notin A'$.

We now use thin indifferences to add $c$ to the choice set. Let $A'':=A'\cup m(N_{\theta_1})\cup\{c\}$. By local IC at $\theta_1$ and  $u(b,\theta_1)\geq u(c,\theta_1)$, we have
$
b\in\argmax_{a'\in A''}u(a',\theta_1).
$
As $b\in A''\subset R$, we have $\theta_0\in\sopt_b^{A''}$. By thin indifferences, there is then some $\theta_2\in N_{\theta_1}\cap\sopt_b^{A''}$. Local IC at $\theta_1$ gives
$
u(m(\theta_2),\theta_2)\geq u(b,\theta_2),
$
and since $m(\theta_2)\in A''$, we get $m(\theta_2)=b$. As $\theta_2\in\sopt_b^{A'\cup\{c\}}$, we obtain
$
b\in m\left(\sopt_b^{A'\cup\{c\}}\right),
$
contradicting the maximality of $A'$, since $c\notin A'$. Thus $m$ is IC.
\end{proof}

\begin{proof}[Proof that statement \ref{cIC:2} $\implies$ statement \ref{cIC:1}.]
We prove the contrapositive: a violation of either property on a finite choice set yields a locally IC mechanism with finite range that is not IC. In both constructions below, every type other than the recipients of $a$ is assigned a maximizer of $B\setminus\{a\}$ at its own type; any two such types therefore satisfy the inequalities in \eqref{e:lIC} for one another, as do any two recipients of $a$. Local IC thus only needs to be verified for comparisons between recipients of $a$ and other types.

Suppose first that connected choice fails: there are a finite $B\subset A$ and $a\in B$ such that $\sopt_a^B$ can be partitioned into two nonempty separated sets $\Theta_1$ and $\Theta_2$. Note that $|B|\geq 2$, because $\sopt_a^{\{a\}}=\Theta$ is connected. Consider the mechanism $m(\theta)=a$ for $\theta\in\Theta_1$ and $m(\theta)\in\argmax_{a'\in B\setminus\{a\}}u(a',\theta)$ otherwise. It is not IC: types in $\Theta_2$ strictly prefer $a$, which they obtain by reporting into $\Theta_1$. We now show that it is, however, locally IC. For $\theta\in\Theta_1$, set $N_\theta:=\Theta\setminus\cl_{\Theta}\Theta_2$; every type in $N_\theta\setminus \Theta_1$ lies outside $\sopt_a^B$ and so finds some element of $B\setminus\{a\}$ weakly better than $a$. For $\theta\in\Theta_2$, set $N_\theta:=\Theta\setminus\cl_{\Theta}\Theta_1$, which contains no recipient of $a$. For every remaining $\theta$, set $N_\theta:=\Theta$; such a type, being outside $\sopt_a^B$, does not strictly prefer $a$ to $m(\theta)$.

Suppose next that thin indifferences fails: there are a finite $B\subset A$ and $a\in B$ with $\sopt_a^B\neq\emptyset$ and a type $\theta_0\in \wopt_a^B\setminus\cl_{\Theta}\sopt_a^B$. Consider the mechanism $m(\theta_0)=a$ and $m(\theta)\in\argmax_{a'\in B\setminus\{a\}}u(a',\theta)$ for $\theta\neq\theta_0$. It is not IC: types in $\sopt_a^B\neq\emptyset$ strictly prefer $a$, which they obtain by reporting $\theta_0$. We now show that it is, however, locally IC.
Set $N_{\theta_0}:=\Theta\setminus\cl_{\Theta}\sopt_a^B$; IC holds for type $\theta_0$ because $\theta_0\in \wopt_a^B$, and no other type in $N_{\theta_0}$ can profitably mimic $\theta_0$ because none lies in $\sopt_a^B$. For $\theta\in\sopt_a^B$, set $N_\theta:=\sopt_a^B$, which is open (a finite intersection of strict-preference sets) and excludes $\theta_0$. For every remaining $\theta$, set $N_\theta:=\Theta$; such a type, being outside $\sopt_a^B$, strictly prefers neither $a$ nor any other assignment to its own.
(If singleton sets of types are closed---as when $\Theta\subset\Reals^n$---then the neighborhoods $\Theta\setminus\{\theta_0\}$ can replace $\sopt_a^B$ above, so this direction does not require the openness assumption.)
\end{proof}

\bigskip

Finally, we substantiate the claim, made in the discussion following \autoref{prop:connected-IC}, that statement \ref{cIC:1} implies statement \ref{cIC:2} for arbitrary $u$ when thin indifferences is strengthened to require $\wopt_a^B\subset\cl_{\Theta}\sopt_a^B$ whenever no single $c\in B\setminus\{a\}$ satisfies $u(c,\cdot)\geq u(a,\cdot)$. Modify the first proof as follows, dropping its second step. After payoff-equivalent actions are collapsed, everywhere-weak-dominance is a partial order on $R$, and any action that dominates a member of $M$ is itself in $M$; as $R$ is finite, $M$ therefore contains an action $b$ that is dominated by no $c\in R\setminus\{b\}$, i.e., for every such $c$ some type strictly prefers $b$ to $c$. The strengthened condition thus applies to $(R,b)$, and since $\hat\theta\in \wopt_b^{R}$, the neighborhood $N_{\hat\theta}$ contains some $\theta_0\in \sopt_b^{R}$. Local IC at $\hat\theta$ gives $m(\theta_0)\notin M$, so $m(\theta_0)\neq b$. This is exactly the configuration reached in the third step of the first proof, and the argument concludes identically.

\subsection{Proof of \autoref{rem:TI-RI}}
\label{proof:TI-RI}

\paragraph{Neither condition implies the other.}
For segment-thin without thin indifferences, let $\Theta=[0,1]$, $A=\{a,b,c\}$, $u(a,\cdot)=0$, $u(b,\theta)=-2\theta$, and $u(c,\theta)=1$ if $\theta\in(0,1/2]$ while $u(c,\theta)=-1$ otherwise. Segment-thin indifferences holds because the only indifference is between $a$ and $b$ at $\theta=0$, and $u(a,\cdot)>u(b,\cdot)$ on $(0,1]$. Thin indifferences fails because $0\in \wopt_a^{A}$, yet $\sopt_a^{A}=(1/2,1]$, whose closure excludes $0$. 

For thin without segment-thin indifferences, let $\Theta=\Reals^2$, $A=\{a,b\}$, $u(a,\theta)=\theta_{(1)}^2-\theta_{(2)}$, and $u(b,\cdot)=0$. Thin indifferences holds because every indifferent type $(t,t^2)$ is the limit of the types $(t,t^2\pm\epsilon)$, which strictly prefer $b$ and $a$, respectively. Segment-thin indifferences fails because $\theta'=(-3,4)$ strictly prefers $a$ and $\theta''=(2,4)$ is indifferent, yet all types $(x,4)$ with $|x|<2$ strictly prefer $b$.

\paragraph{Under the openness assumption, segment-thin indifferences implies thin indifferences.}
Assume every strict-preference set is open and $u$ has segment-thin indifferences. Fix a finite choice set $B$, $a\in B$ with some $\eta\in \sopt_a^{B}$, and $\theta_0\in \wopt_a^{B}$. If $\theta_0=\eta$ there is nothing to prove, so assume $\theta_0\neq \eta$. Parametrize $\ell(\eta,\theta_0)$ by $[0,1]$ and, for $b\in B\setminus\{a\}$, let $f_b:=u(a,\cdot)-u(b,\cdot)$ along this segment. Each $f_b$ satisfies $f_b(0)>0$ and $f_b(1)\geq 0$, and the sets $\{f_b>0\}$ and $\{f_b<0\}$ are open in $[0,1]$. We claim $f_b\geq 0$ on $[0,1]$. If $f_b(s)<0$ for some $s$, let $q:=\inf\{t\geq s: f_b(t)\geq 0\}$; openness of the two strict sets implies $f_b<0$ on $[s,q)$ and $f_b(q)=0$. But segment-thin indifferences applied to the pair at parameters $0$ and $q$ requires points at which $a$ is strictly preferred to $b$ accumulating at $q$ from within $[0,q]$, and there are none in $[s,q)$---a contradiction. Next, $\{f_b>0\}$ is open and dense in $[0,1]$: every zero of $f_b$ is a limit of strict points, by segment-thin indifferences applied to the pair at parameters $0$ and that zero. A finite intersection of open dense subsets of $[0,1]$ is dense, so points at which $a$ is strictly preferred to every $b\in B\setminus\{a\}$ accumulate at $1$. Such points lie in $\sopt_a^{B}$, whence $\theta_0\in\cl_{\Theta}\sopt_a^{B}$.

\paragraph{Given convex choice, segment-thin indifferences holds if and only if thin indifferences holds on every line segment.}
For a line segment $\ell\subset\Theta$ and $B\subset A$, $a\in B$, write $\sopt_a^{B,\ell}$ and $\wopt_a^{B,\ell}$ for the sets from \autoref{sec:CC} defined relative to $\ell$, with closures also taken in $\ell$.
If segment-thin indifferences fails for some $a',a''$ and $\theta',\theta''$, let $\ell:=\ell(\theta',\theta'')$ and $B:=\{a',a''\}$. Then $\theta'\in \sopt_{a'}^{B,\ell}$ and $\theta''\in \wopt_{a'}^{B,\ell}$, while a sequence in $\sopt_{a'}^{B,\ell}$ converging to $\theta''$ is precisely what the failure denies; so thin indifferences fails on $\ell$.

Conversely, suppose $u$ has convex choice and segment-thin indifferences. Fix a segment $\ell$, a finite choice set $B$, $a\in B$ with some $\eta\in \sopt_a^{B,\ell}$, and any $\theta_0\in \wopt_a^{B,\ell}$. If $\theta_0=\eta$ there is nothing to prove, so assume $\theta_0\neq\eta$. It suffices to show that $a$ is strictly preferred to each $b\in B\setminus\{a\}$ at every point strictly between $\eta$ and $\theta_0$; the open segment between $\eta$ and $\theta_0$ then lies in $\sopt_a^{B,\ell}$, whence $\theta_0\in\cl_{\ell}\sopt_a^{B,\ell}$. If $u(a,\theta_0)>u(b,\theta_0)$, this follows from the convexity of $\sopt_a^{\{a,b\}}$, which contains $\eta$ and $\theta_0$. If $u(a,\theta_0)=u(b,\theta_0)$, segment-thin indifferences supplies types $\theta_n\in\ell(\eta,\theta_0)$ converging to $\theta_0$ with $u(a,\theta_n)>u(b,\theta_n)$. Any point strictly between $\eta$ and $\theta_0$ lies in $\ell(\eta,\theta_n)$ for large $n$, so the convexity of $\sopt_a^{\{a,b\}}$, which contains $\eta$ and $\theta_n$, yields the conclusion.%
\hfill \qed
\subsection{Proof of \autoref{prop:weakmon}}
\label{proof:weakmon}

Let $\upsilon:\Theta\to Y$ be a weakly monotone allocation rule with finite range. Since $\Theta\subset\mathbb{R}^n$ is convex and $\tilde{u}(y,\theta)$ is continuous in $\theta$,  Theorem 5 in \citet[p.~380]{BMN:17} implies that it is sufficient to show that $\upsilon$ is ``line implementable'' \citep[p.~373]{BMN:17}.

Towards contradiction, suppose there is a line segment such that the restriction of $\upsilon$ to this line segment is not implementable. We can label types on this line segment as $[0,1]$ corresponding to the convex combinations of its end points. It follows from \citet{rochet87} that cyclical monotonicity is violated: there is a finite set $\{\theta_1,\ldots,\theta_K\}\subset[0,1]$ such that
\begin{equation}
\label{e:cycmon}
\sum_{i=1}^K [\tilde u(\upsilon(\theta_{i+1}),\theta_i)-\tilde u(\upsilon(\theta_i),\theta_i)]>0,
\end{equation}
where $\theta_{K+1}:= \theta_1$.

For $y,y'\in\upsilon(\Theta)$, let $\tilde D_{y,y'}$ denote the restriction of $\tilde u(y,\cdot)-\tilde u(y',\cdot)$ to $[0,1]$, which is monotonic by convex choice (as noted after \autoref{prop:weakmon}). Together with $\tilde D_{y,y''}=\tilde D_{y,y'}+\tilde D_{y',y''}$, monotonicity implies that the relation ``$\tilde D_{y,y'}$ is increasing'' on the finite set $\upsilon(\Theta)$ is complete and transitive, so we can enumerate $\upsilon(\Theta)=\{y_1,\ldots,y_J\}$ such that $\tilde D_{y_k,y_l}$ is increasing whenever $k>l$. 

We claim that if $\upsilon(\theta)=y_k$ and $\upsilon(\theta')=y_l$ with $k>l$, then $\tilde D_{y_{j},y_{j-1}}(\theta)\geq \tilde D_{y_{j},y_{j-1}}(\theta')$ for all $j\in\{l+1,\ldots,k\}$. The claim is immediate if $\theta\geq \theta'$. If $\theta<\theta'$, then weak monotonicity gives $\tilde D_{y_k,y_l}(\theta)\geq \tilde D_{y_k,y_l}(\theta')$, so the increasing function $\tilde D_{y_k,y_l}=\sum_{j=l+1}^k \tilde D_{y_j,y_{j-1}}$ is constant on $[\theta,\theta']$, and hence so is each of its increasing summands.

Now let $k_i$ be the index with $\upsilon(\theta_i)=y_{k_i}$, so that $k_{K+1}=k_1$. The left-hand side of \eqref{e:cycmon} equals
\[
-\sum_{i=1}^K \tilde D_{y_{k_i},y_{k_{i+1}}}(\theta_i)=-\sum_{j=2}^J\Big[\sum_{i:\,k_i\geq j} \tilde D_{y_j,y_{j-1}}(\theta_i)-\sum_{i:\,k_{i+1}\geq j} \tilde D_{y_j,y_{j-1}}(\theta_i)\Big],
\]
where the equality follows from $\tilde D_{y_k,y_l}=\sum_{j=2}^{k}\tilde D_{y_j,y_{j-1}}-\sum_{j=2}^{l}\tilde D_{y_j,y_{j-1}}$ by exchanging the order of summation. Fix $j$. Since $k_{K+1}=k_1$, the two sums over $i$ inside the display's bracket have equally many terms.
The terms of the second sum with $k_i \geq j$ also appear in the first sum. By the claim in the previous paragraph, every term $\tilde D_{y_j,y_{j-1}}(\theta_i)$ with $k_i\geq j$ is at least as large as every term with $k_i<j$, so the first sum is at least the second. Hence the left-hand side of \eqref{e:cycmon} is at most $0$, a contradiction.
\hfill \qedsymbol

\subsection{Norm-Based Preferences}
\label{proof:norms}

\begin{proof}[Proof of \autoref{rem:norms}]
It suffices to prove that $\lVert\cdot\rVert$ is an inner-product norm, as standard arguments then imply that
it is weighted Euclidean.

Accordingly, suppose to contradiction that $\lVert\cdot\rVert$ is not an inner-product norm yet there is DSCD. For any pair of actions $a,a'\in A$, the set of indifferent types is the set of types that are equidistant to $a$ and $a'$, i.e., $\{\theta\in\Theta: \lVert a-\theta\rVert = \lVert a'- \theta\rVert\}$. DSCD implies that this set is convex. However, as $\lVert \cdot \rVert$ is not an inner-product norm, it follows from \citet[][Corollary 1.3]{PK:74} that for some $a\in \Reals^n$, the set $\{\theta\in\R^n: \lVert a-\theta\rVert=\lVert -a-\theta\rVert\}$ is not convex. Hence, there exist $\theta,\theta'\in \R^n$ and $\lambda \in(0,1)$ such that $\theta$ and $\theta'$ are equidistant to $a$ and $-a$, but $\theta'':=\lambda \theta + (1-\lambda) \theta'$ is closer to, say, $a$. Let $\theta_0$ be in the interior of $A\cap \Theta$. We can then choose $\e>0$ sufficiently small such that $$\theta_0+\e \theta, \theta_0+\e\theta',\theta_0+\e\theta'',\theta_0+\e a,\theta_0-\e a\in A\cap\Theta$$ and obtain
\begin{align*}
u(\theta_0+\e a,\theta_0+\e\theta'')&= -l(\e \lVert a-\theta''\rVert)>-l(\e \lVert -a-\theta''\rVert) = u(\theta_0-\e a,\theta_0+\e\theta''),\\
u(\theta_0+\e a,\theta_0+\e\theta)&=u(\theta_0-\e a,\theta_0+\e\theta),\\
u(\theta_0+\e a,\theta_0+\e\theta')&=u(\theta_0-\e a,\theta_0+\e\theta'),
\end{align*}
contradicting that the set of indifferent types is convex.
\end{proof}

\begin{remark}
\label{rem:quadratic}
Let $\Theta=\Reals^n$ with $n\geq 2$, and consider norm-based preferences (\autoref{rem:norms}) over a set of outcomes $X\subset\Reals^n$ with nonempty interior. Expected utility over $\Delta X$ has DSCD if and only if the norm is weighted Euclidean and $l(d)=cd^2+k$ for some $c>0$ and $k\in\Reals$.
\end{remark}

\begin{proof}[Proof.]
The ``if'' direction holds because utility differences between lotteries are then affine in $\theta$. For the ``only if'' direction, \autoref{rem:norms} applied to sure outcomes implies that the norm is weighted Euclidean; after a linear change of variables and a translation, it is Euclidean and $X$ contains a ball of radius $r>0$ around $0$. Let $L(s):=-l(\sqrt{s})$. For $h\in(0,r)$, compare the equiprobable lotteries on $\{\pm he_1\}$ and on $\{\pm he_2\}$, where $e_i$ is the $i$-th standard basis vector. At type $te_1$, the expected utility difference is $D(t):=\tfrac12L\big((t-h)^2\big)+\tfrac12L\big((t+h)^2\big)-L(t^2+h^2)$. As $D$ is even with $D(0)=0$, and DSC makes its sign monotonic in $t$, we have $D\equiv0$. Substituting $t=(\sqrt{s}+\sqrt{s'})/2$ and $h=\lvert\sqrt{s}-\sqrt{s'}\rvert/2$ yields $L(s)+L(s')=2L\big((s+s')/2\big)$ whenever $s,s'\geq0$ and $\lvert s-s'\rvert<4r^2$. Monotonicity makes $L$ locally, hence globally, affine on $[0,\infty)$, yielding $l(d)=cd^2+k$ for some $c>0$ and $k$.
\end{proof}

\subsection{Proof of \autoref{prop:CC-DSCD}}
\label{proof:CC-DSCD}
\paragraph{Definitions.} Recall that $D_{a,a'}(\theta)\equiv u(a,\theta)-u(a',\theta)$.  We say that a function $f:\Theta\rightarrow \R$ \emph{strictly violates DSC} if for all $\alpha\in \R^n\setminus\{0\}$, there are $\theta,\theta',\theta''\in \Theta \text{ with } \alpha\cdot \theta< \alpha\cdot \theta' < \alpha\cdot \theta''$ such that
\[\min\{f(\theta), f(\theta'')\}>0>f(\theta') \text{ \ or \ } \max\{f(\theta), f(\theta'')\}<0<f(\theta').\]
The utility function $u$ \emph{strictly violates DSCD} if there are $a,a'\in A$ such that $D_{a,a'}(\theta)$ strictly violates DSC.

Note that a strict violation of DSC(D) is slightly stronger than just violating DSC(D). To illustrate, consider $\Theta=[-1,1]$ and $f(\theta)=|\theta|$. This function is not DSC because its sign is not monotonic. But it does not strictly violate DSC. If, instead, $f(0)<0$ (maintaining  $f(\theta)>0$ for $\theta\neq 0$), then it does strictly violate DSC. The violation of DSC in \autoref{fig:DSC3} is also a strict violation.

\paragraph{Proof.} For the proposition's first statement, suppose $u$ has DSCD. It is sufficient to show convex choice for all binary choice sets.  So consider any $A'=\{a',a''\}\subset A$. Let $\theta',\theta''\in\{\theta: \{a'\}= \argmax_{a\in A'} u(a,\theta)\}$. We have $D_{a',a''}(\theta')> 0$ and $D_{a',a''}(\theta'')> 0$. Since $u$ has DSCD, there is $\alpha\in \R^n\setminus\{0\}$ such that $D_{a',a''}(\cdot)$ is DSC in direction $\alpha$. This implies $D_{a',a''}(\theta)> 0$ for all $\theta\in\ell(\theta',\theta'')$, as required.

For the proposition's second statement, suppose $u$ has convex choice. Fix arbitrary $a',a''\in A$ and let
\begin{align*}
\Theta'&:= \{ \theta: D_{a',a''}(\theta)>0\},\\
\Theta''&:= \{ \theta: D_{a',a''}(\theta)<0\}.
\end{align*}
It is straightforward that there is not a strict violation of DSCD if either of these sets is empty.
So assume $\Theta'\neq \emptyset$ and $\Theta''\neq \emptyset$. Then $\Theta'$ and $\Theta''$ are nonempty, convex (by convex choice), and disjoint sets. A standard separating hyperplane theorem implies that there is $\alpha\in\R^n\setminus\{0\}$ such that $\alpha\cdot \theta' \le \alpha\cdot \theta''$ for all $\theta'\in \Theta'$ and $\theta''\in \Theta''$. Hence, there are no $\theta,\theta',\theta''\in \Theta$ with $\alpha\cdot \theta<\alpha\cdot \theta'<\alpha \cdot \theta''$ such that either $\min\{D_{a',a''}(\theta), D_{a',a''}(\theta'')\}>0>D_{a',a''}(\theta')$ or $\max\{D_{a',a''}(\theta), D_{a',a''}(\theta'')\}<0<D_{a',a''}(\theta')$, i.e., there is not a strict violation of DSCD. \hfill \qed

\subsection{Proof of \autoref{prop:DSCDstar}}

For this proof, we fix an arbitrary action $a_* \in A$ and work with the function $$f(a,\theta):=u(a,\theta)-u(a_*,\theta).$$ Since we are interested in representations of $u$, we say that a function $\tilde f$ is a \emph{representative} of $f$ if there is a representation $\tilde u$ of $u$ such that $\tilde f(a,\theta)=\tilde u(a,\theta)-\tilde u(a_*,\theta)$. Note that any type-dependent positive linear transformation of $f$ (i.e., $g(\theta)f(a,\theta)$ with $g(\cdot)>0$) is a representative of $f$.
We say that $f$ is \emph{DSC-preserving} if for all finite index sets $I$, $\{a_i\}_{i\in I}\subset A$, and $\{\lambda_i\}_{i\in I}\subset \Reals$, it holds that $\sum_{i\in I}\lambda_i f(a_i,\cdot)$ is DSC.

\begin{proof}[Proof of \autoref{prop:DSCDstar}]

An affine representation implies DSCD because affine utility differences are DSC and representations preserve their signs. For the other implication, assume DSCD.

Let $A_*:=A\setminus\{a_*\}$.  
Analogous to \citet[proof of Theorem 2]{KLR:23}, for every $\mu:A_*\rightarrow \R$ with finite support $A_\mu$, there exist probability distributions $P,Q$ on $A_\mu \cup \{a_*\}$ such that $\sum_{a\in A_\mu} f(a,\theta) \mu(a)$ is a multiple of 
\begin{equation}
    \label{e:decomposition}
    \sum_{a\in A_\mu \cup \{a_*\}} u(a,\theta) P(a)-\sum_{a\in A_\mu \cup \{a_*\}} u(a,\theta) Q(a).
\end{equation} Since the environment is convex and $u$ has DSCD, the function \eqref{e:decomposition} is DSC, which implies that $\sum_{a\in A_{\mu}} f(a,\theta) \mu(a)$ is also DSC. Hence, $f(a,\theta)$ is DSC-preserving. \autoref{p:DSCstar} below implies that $f$ has an affine representative or is one-dimensional. It follows that $u$ has an affine representation or is one-dimensional. 
\end{proof}

Accordingly, the key to \autoref{prop:DSCDstar} is the following result.

\begin{proposition}\label{p:DSCstar}
    Assume $\Theta=\R^n$ and $f(a,\theta)$ is DSC-preserving and differentiable in $\theta$.
    If
    no type is totally indifferent,
    then $f$ is either one-dimensional or has an affine representative.
\end{proposition}

We first establish a normalization lemma on convex type spaces, which will also be used in the proof of \autoref{prop:DSCDstar-strict-dominance}.
We work in the affine hull of \(\Theta\), identified with a Euclidean space.
Thus the relative interior of \(\Theta\), which we denote \(\ri\Theta\), is its (nonempty) interior in this space, and gradients below are taken in this space.

\begin{lemma}\label{lemma:affine_representative}
Let $V$ be a vector space of differentiable functions on a convex set $\Theta$, all of which are DSC. If linearly independent $f_0,f_1\in V$ vanish at $\theta_0\in \ri\Theta$ and $\nabla f_0(\theta_0)\neq0$, then there is $\beta:\Theta\to(0,\infty)$ such that $\beta g$ is affine for every $g\in V$.
\end{lemma}

\begin{proof}
First, if $g\in V$ and $g(z)=f_0(z)=0$ at $z\in\ri\Theta$, then
\begin{equation}\label{e:DSC_two_zeros}
g(\theta)\nabla f_0(z)\cdot(\theta-z)
=f_0(\theta)\nabla g(z)\cdot(\theta-z)
\qquad \text{for every }\theta\in\Theta.
\end{equation}
Indeed, for fixed $\theta$, $g(\theta)f_0(\cdot)-f_0(\theta)g(\cdot)$ is DSC and vanishes at both endpoints of the segment from $z$ to $\theta$. Hence it vanishes on the segment, and its directional derivative at $z$ in direction $\theta-z$ is zero.
Applied with $g=f_1$ at $\theta_0$, \eqref{e:DSC_two_zeros} shows that $\nabla f_0(\theta_0)$ and $\nabla f_1(\theta_0)$ are independent: proportional gradients would make the functions proportional off a hyperplane, and hence everywhere by continuity, contradicting their linear independence.

Next, for $i\in\{0,1\}$, we claim 
\begin{align}\label{eq:f_nablaf}
\sign f_i(\theta)
=
\sign\!\bigl(\nabla f_i(\theta_0)\cdot(\theta-\theta_0)\bigr)
\qquad\text{for every }\theta\in\Theta.    
\end{align}
Indeed, fix $i$ and choose a direction \(\alpha\neq0\) in which \(f_i\) is DSC. DSC implies that \(f_i\) vanishes on the hyperplane through \(\theta_0\) perpendicular to \(\alpha\), so the gradient \(\nabla f_i(\theta_0)\neq 0\) is perpendicular to this hyperplane; as \(f_i\geq 0\) where \(\alpha\cdot(\theta-\theta_0)>0\), the gradient is a positive multiple of \(\alpha\). 
Fix \(\theta\in\Theta\) with \(\alpha\cdot(\theta-\theta_0)>0\). By differentiability,
\[
f_i\bigl(\theta_0+t(\theta-\theta_0)\bigr)
=
t\,\nabla f_i(\theta_0)\cdot(\theta-\theta_0)+o(t)>0
\]for sufficiently small \(t>0\). DSC then implies \(f_i(\theta)>0\). The argument for \(\alpha\cdot(\theta-\theta_0)<0\) is analogous, and both sides of the claimed identity vanish when \(\alpha\cdot(\theta-\theta_0)=0\).

We now define a common rescaling that makes $f_0$ and $f_1$ affine. Define
\[
\tilde f_1(\theta):=\nabla f_1(\theta_0)\cdot(\theta-\theta_0),
\]
and, for every $g\in V$, set
\begin{equation}\label{e:DSC_ratio}
\tilde g(\theta):=\frac{\tilde f_1(\theta)}{f_1(\theta)}g(\theta)
\qquad\text{where }f_1(\theta)\neq0.
\end{equation}
The multiplier is positive by \eqref{eq:f_nablaf}. 
Equation \eqref{e:DSC_two_zeros} at $\theta_0$, with $g=f_1$, gives
\[
\tilde f_0(\theta)=\nabla f_0(\theta_0)\cdot(\theta-\theta_0),
\]
so $\tilde f_0$ extends affinely to $\Theta$.

Choose $\theta_1\in\ri\Theta$ with $f_0(\theta_1)=0\neq f_1(\theta_1)$, which is possible by \eqref{eq:f_nablaf} and independence of the gradients. We next obtain affine extensions for functions vanishing at $\theta_1$. Note that differentiating $f_0\tilde f_1=f_1\tilde f_0$ at $\theta_1$,
where $f_0(\theta_1)=\tilde f_0(\theta_1)=0$, gives
\[
\frac{\tilde f_1(\theta_1)}{f_1(\theta_1)}
\nabla f_0(\theta_1)=\nabla\tilde f_0(\theta_1).
\]
Since $\tilde f_0(\theta)=\nabla \tilde f_0(\theta_1)\cdot(\theta-\theta_1)$, this yields
\begin{align}\label{eq:tilde_f_0_theta}
\frac{\tilde f_1(\theta_1)}{f_1(\theta_1)}
\nabla f_0(\theta_1)\cdot(\theta-\theta_1)=\tilde f_0(\theta).    
\end{align}
Now, for any $g\in V$ vanishing at $\theta_1$ and any
$\theta\in\Theta$ with $f_0(\theta)f_1(\theta)\neq0$,
\begin{align}\label{e:DSC_affine_zero}
    \tilde g(\theta)     &=\frac{\tilde f_0(\theta)}{f_0(\theta)}g(\theta) \nonumber \\
    &=\frac{\tilde f_1(\theta_1)}{f_1(\theta_1)}
\frac{g(\theta)}{f_0(\theta)}
\nabla f_0(\theta_1)\cdot(\theta-\theta_1)\nonumber \\
&=
\frac{\tilde f_1(\theta_1)}{f_1(\theta_1)}
\nabla g(\theta_1)\cdot(\theta-\theta_1),
\end{align}
where the first equality follows from \eqref{e:DSC_ratio}, the second is \eqref{eq:tilde_f_0_theta}, and the last is \eqref{e:DSC_two_zeros}.
By continuity, the identity
\[
\tilde g(\theta)
=\frac{\tilde f_1(\theta_1)}{f_1(\theta_1)}
\nabla g(\theta_1)\cdot(\theta-\theta_1)
\]
holds wherever $f_1(\theta)\neq0$. We use this affine formula to define $\tilde g$ at the remaining points of $\Theta$.

We next argue that $\tilde g$ is affine even if $g(\theta_1)\neq 0$. For any $g\in V$, the function
\[
g(\cdot)-\frac{g(\theta_1)}{f_1(\theta_1)}f_1(\cdot)
\]
vanishes at $\theta_1$, so its rescaling is affine by the preceding argument. Adding back the corresponding multiple of $\tilde f_1$ gives an affine extension of $\tilde g$.

It remains to verify these affine functions are obtained by a common positive rescaling even on $\{f_1=0\}$.
For nonzero $g$ vanishing at $\theta_1$, \eqref{e:DSC_ratio} and \eqref{e:DSC_affine_zero} imply $\nabla g(\theta_1)\neq0$. Since the factor in \eqref{e:DSC_affine_zero} is positive, the sign argument proving \eqref{eq:f_nablaf}, applied to $g$ at $\theta_1$, shows that $\tilde g$ and $g$ have the same sign. For general $g$, subtract $[g(\theta_1)/f_1(\theta_1)]f_1$ to make it vanish at $\theta_1$. By linearity of $g\mapsto\tilde g$, this subtracts the same multiple of $\tilde f_1$ from $\tilde g$. Since $f_1=\tilde f_1=0$ on $\{f_1=0\}$, neither value changes, so signs agree there as well. Finally, linearity and preservation of zeros imply that the ratios $\tilde g(\theta)/g(\theta)$, whenever defined, are independent of $g$. Their positivity gives the required common rescaling.
\end{proof}

\begin{proof}[Proof of \autoref{p:DSCstar}]
We first identify the one-dimensional case. Let $V:=\Span\{f(a,\cdot):a\in A\}$. For each member, choose a direction in which it is DSC. If all these directions can be chosen to be parallel, the signs of all functions in $V$ depend only on a common projection $\alpha\cdot\theta$, so $f$ is one-dimensional.

Otherwise, there are $g_0,g_1\in V$ for which the DSC directions cannot be chosen to be parallel; in particular, $g_0$ and $g_1$ are independent. Continuity and DSC imply that each zero set $\{g_i(\cdot)=0\}$ has a boundary hyperplane perpendicular to its DSC direction.
These hyperplanes are not parallel and, since $\Theta=\mathbb R^n$, intersect at some $\bar\theta$.
No total indifference provides an action $a'$ such that
$p(\cdot):=f(a',\cdot)$ is nonzero at $\bar\theta$.\footnote{Under the weaker assumption in \autoref{fn:totalindiff},
if $\bar\theta$ is totally indifferent, choose $q\in V$
with $\nabla q(\bar\theta)\neq0$, and choose a member
of $\{g_0,g_1\}$ linearly independent of $q$. Both functions vanish at $\bar\theta$, so \autoref{lemma:affine_representative} gives the conclusion. We may therefore assume that $\bar\theta$ is not totally indifferent.}
Since $g_0,g_1$ vanish there and $p$ does not,
$g_0,g_1,p$ are linearly independent.

We now construct the two functions required by \autoref{lemma:affine_representative}. Near $\bar\theta$, $p$ is nonzero and $g_0/p$ is nonconstant, because $\bar\theta$ lies on the boundary of the zero set of $g_0$. The mean value theorem therefore gives $z$ with $\nabla(g_0/p)(z)\neq0$. Set
\[
h(\cdot):=g_0(\cdot)-\frac{g_0(z)}{p(z)}p(\cdot).
\]
Then $h(z)=0$ and $\nabla h(z)=p(z)\nabla(g_0/p)(z)\neq0$. The function $g_1-\frac{g_1(z)}{p(z)}p$ also vanishes at $z$ and is linearly independent of $h$, since $g_0,g_1,p$ are linearly independent. This verifies the hypotheses of \autoref{lemma:affine_representative}. The lemma yields the required affine representative of $f$.
\end{proof}

\subsection{Proof of \autoref{prop:DSCDstar-strict-dominance}}

An affine representation implies DSCD. For the converse, assume DSCD, let $a^*$ strictly dominate $a_*$, and define
\[
f(a,\theta):=\frac{u(a,\theta)-u(a_*,\theta)}{u(a^*,\theta)-u(a_*,\theta)}.
\]
Then $f(a_*,\cdot)=0$ and $f(a^*,\cdot)=1$. Let $V:=\Span\{f(a,\cdot):a\in A\}$. By convexity of the environment and the decomposition
\eqref{e:decomposition}, every member of $V$ is DSC; dividing by the positive function $u(a^*,\cdot)-u(a_*,\cdot)$ does not affect signs. Minimal richness implies $\dim V\geq3$.

Accordingly, we can choose a nonconstant $q\in V$. The mean value theorem gives $z\in\ri\Theta$ with $\nabla q(z)\neq0$. Since $1\in V$, the function $h:=q-q(z)$ belongs to $V$ and satisfies $h(z)=0$ and $\nabla h(z)\neq0$. Moreover, as $\dim V\geq3$, there is $b'\in A$ with $f(b',\cdot)\notin\Span\{1,h\}$; then $f(b',\cdot)-f(b',z)$ belongs to $V$, vanishes at $z$, and is linearly independent of $h$. Hence \autoref{lemma:affine_representative} yields $\beta:\Theta\to(0,\infty)$ such that $\beta f$ is affine in $\theta$. As $\beta(\theta)/[u(a^*,\theta)-u(a_*,\theta)]>0$, $\beta f$ is a representation of $u$, and hence the required affine representation. \hfill\qed

\section{Examples}
\label{app:examples}

\begin{example}
\label{eg:intervalvsconvex}
Ignoring indifferences, convex choice is equivalent to ``interval choice on every line segment'' (viewing each segment as a one-dimensional type space). We show that this is neither stronger nor weaker than interval choice on the full type space when $\Theta$ is multidimensional. Consider the following two examples with $\Theta=[0,1]^2$ and $A=\{a',a''\}$:
    \begin{center}
    \begin{tikzpicture}[scale=2.5]
      \draw[thick] (0,0) rectangle (1,1);
      \fill[blue!20!white] (0,0) -- (1,1) -- (1,0) -- cycle;
      \fill[red!10!white] (0,0) -- (0,1) -- (1,1) -- cycle;
      \coordinate (midpoint) at (.5,0.5);
      \draw[blue, thick] (0,0) -- (1,1);  
      \node at (0,-0.1) {0};
      \node at (1,-0.1) {1};
      \node at (-0.1,0) {0};
      \node at (-0.1,1) {1};
      \node[red] at (0.3,0.75) {$a'' \succ a'$};
      \node[blue] at (0.7,0.25) {$a' \succ a''$};
    \end{tikzpicture}
    \hspace{.5in}
    \begin{tikzpicture}[scale=2.5]
      \draw[thick] (0,0) rectangle (1,1);
      
      \fill[blue!20!white] (0,1) -- plot[domain=0:1, samples=100] (\x, {1-sqrt(\x)}) -- (1,1) -- cycle;
      \fill[red!10!white] (0,0) -- plot[domain=0:1, samples=100] (\x, {1-sqrt(\x)}) -- (1,0) -- cycle;
      
      \draw[domain=0:1, samples=100, blue, thick] plot (\x, {1-sqrt(\x)});
      
      \node at (0,-0.1) {0};
      \node at (1,-0.1) {1};
      
      \node at (-0.1,0) {0};
      \node at (-0.1,1) {1};
      
      \node[blue] at (0.65,0.7) {$a' \succ a''$};
      \node[red] at (0.3,0.15) {$a'' \succ a'$};
    \end{tikzpicture}
        \end{center}
On the left, there is convex but not interval choice; on the right, there is interval but not convex choice.
\end{example}

\begin{example}
\label{eg:Carroll-linear-loss}
We show that a utility function can satisfy the conditions of \autoref{prop:CC-IC} and yet preclude obtaining a convex type space by identifying types with utility vectors over outcomes. Let $X=\{0,1/2,1\}$ be a set of three outcomes, $A$ be the lotteries over $X$, and for $\Theta=[0,1]$ let $u(a,\theta)=\sum_{x\in X}a_{x}\left(-|x-\theta|\right)$, i.e., expected utility with a linear-loss Bernoulli utility. Let us verify that this specification has convex choice. For lotteries $a,a'\in A$, the difference $g(\theta):=u(a,\theta)-u(a',\theta)$ is piecewise linear with a kink only at $\theta=1/2$, and $g(1)=-g(0)$ because $\sum_{x}(a_{x}-a'_{x})=0$. If the set $\{\theta:g(\theta)>0\}$ were not an interval, then $g$ would be strictly positive at both endpoints of $[0,1]$ or strictly negative at both, contradicting $g(1)=-g(0)$; so convex choice holds. A similar argument verifies segment-thin indifferences.

However, since preferences over lotteries pin down each type's utility vector over $X$ up to a positive affine transformation, no choice of these utility vectors yields a convex set. To show that, suppose that functions $\lambda(\cdot)>0$ and $\mu(\cdot)$ are such that $$K:=\left\{\lambda(\theta)\left(-|x-\theta|\right)_{x\in X}+\mu(\theta)\mathbf{1}:\theta\in[0,1]\right\}$$ is convex. The linear map $(w_0,w_{1/2},w_1) \mapsto\left(w_{1/2}-w_{0},w_{1}-w_{0}\right)$ sends $K$'s elements for types $0$ and $1$ to $\lambda(0)\left(-1/2,-1\right)$ and $\lambda(1)\left(1/2,1\right)$, which point in opposite directions. Hence some convex combination of the two utility vectors lies in $K$ and is sent to $(0,0)$; it would be the utility vector of a type indifferent among all three outcomes, but there is no such type.
\end{example}

\begin{example}
\label{eg:strictviolation}
If $\Theta \subset \Reals$, then absent indifferences, a violation of DSCD is equivalent to a strict violation of DSCD. But not so in multiple dimensions. Consider the following two examples with $\Theta=[0,1]^2$ and $A=\{a',a''\}$:
    \begin{center}
    \begin{tikzpicture}[scale=2.5]
      \draw[thick] (0,0) rectangle (1,1);
      \fill[blue!20!white] (0,1) -- (1,1) -- (1,0) -- cycle;
      \fill[red!10!white] (0,0) -- (0,1) -- (1,0) -- cycle;
      \coordinate (midpoint) at (.5,0.5);
      \draw[blue, line width=1.5pt] (0,1) -- (midpoint);  
      \draw[red, line width=1.5pt]  (1,0) -- (midpoint);  
      \node at (0,-0.1) {0};
      \node at (1,-0.1) {1};
      \node at (-0.1,0) {0};
      \node at (-0.1,1) {1};
      \node[blue] at (0.7,0.75) {$a' \succ a''$};
      \node[red] at (0.3,0.25) {$a'' \succ a'$};
    \end{tikzpicture}
\hspace{.5in}
    \begin{tikzpicture}[scale=2.5]
      \draw[thick] (0,0) rectangle (1,1);
      \fill[blue!20!white] (0,1) -- (1,1) -- (1,0) -- cycle;
      \fill[red!10!white] (0,0) -- (0,1) -- (1,0) -- cycle;
      \coordinate (A) at (.25,0.75);
      \coordinate (B) at (.75,0.25);
      \draw[blue, line width=1.5pt] (0,1) -- (B);  
      \draw[blue, line width=1.5pt] (1,0) -- (A); 
      \draw[red, line width=1.5pt]  (A) -- (B);  
      \node at (0,-0.1) {0};
      \node at (1,-0.1) {1};
      \node at (-0.1,0) {0};
      \node at (-0.1,1) {1};
      \node[blue] at (0.7,0.75) {$a' \succ a''$};
      \node[red] at (0.3,0.25) {$a'' \succ a'$};
    \end{tikzpicture}
    \end{center}
Both examples have no indifferences, a violation of DSCD, but no strict violation of DSCD. On the left, there is convex choice (so convex choice does not imply DSCD); on the right, convex choice fails (so no strict violation of DSCD does not imply convex choice).
\end{example}

\begin{example}
\label{eg:one-dim-DSCD-no-affine-rep}
We show that when preferences are one-dimensional, there are convex environments in which DSCD utilities do not have an affine representation. Consider $\Theta=[1,2]$, $A=\Delta\left(\{x_0,x_1,x_2\}\right)$, and for any lottery $a\in A$, $u(a,\theta)=\sum_{i=0}^2 a(x_i)\bernoulli(x_i,\theta)$ with 
$\bernoulli(x_0,\theta)=0$, $\bernoulli(x_1,\theta)=1$ and $\bernoulli(x_2,\theta)=\theta^2$. DSCD holds because the utility difference between any pair of lotteries is a linear combination of $1$ and $\theta^2$, whose sign is monotonic on $\Theta$. To see that there is no affine representation, write a lottery as its pair of probabilities on $x_1$ and $x_2$, and consider lotteries $a=(1,0)$, $b=(20/47,12/47)$, $c=(27/47,5/47)$, and $d=(0,17/47)$. Then type $1$ is indifferent between $b$ and $c$; type $3/2$ between $a$ and $b$ and between $c$ and $d$; and type $2$ between $a$ and $c$ and between $b$ and $d$. Suppose $\tilde u$ is an affine representation; we can normalize $\tilde u(b,\cdot)=0$. The indifferences with $b$ imply $\tilde u(c,1)=\tilde u(a,3/2)=\tilde u(d,2)=0$. Thus $\tilde u(c,\theta)=k(\theta-1)$ for some $k\in\Reals$, and the indifferences with $c$ give $\tilde u(a,\theta)=k(2\theta-3)$ and $\tilde u(d,\theta)=k(2-\theta)$. Hence $\tilde u(a,5/3)=\tilde u(d,5/3)=k/3$; but this indifference contradicts $u(d,5/3)=425/423>1=u(a,5/3)$.
\end{example}

\begin{remark}
\label{rem:1d}
Although one-dimensional preferences with DSCD in a convex environment need not have an affine representation, their structure can be characterized using \citepos{KLR:23} results.  If $\Theta$ is compact, then all types' preferences must be a convex combination of two ``extreme'' types', with an ordered weighting. More precisely, there must be a representation \mbox{$\lambda (\alpha \cdot \theta) \overline u(a)+\left(1-\lambda (\alpha \cdot \theta)\right)\underline u(a)$}, where $\alpha\in \Reals^n\setminus\{0\}$ is the direction in which preferences are one-dimensional, $\lambda:\Reals \to [0,1]$ is an increasing function, and $\overline u$ and $\underline u$ represent the preferences of the types that respectively maximize and minimize $\alpha \cdot \theta$.
\end{remark}

\begin{example}
\label{eg:SCD-no-affine-rep}
We show that under \citepos{KLR:23} single-crossing differences (SCD), multidimensional preferences in a convex environment need not have an affine representation, even when the hypotheses of \autoref{prop:DSCDstar} hold. Consider $\Theta=\Reals^2$ and expected utility $u(a,\theta)$ over lotteries $a$ on three outcomes whose utilities for type $\theta$ are $0$, $1$, and $\theta_{(1)}+\theta_{(2)}^3$. Plainly, $u$ is differentiable in $\theta$ and no type is totally indifferent. The utility difference between any pair of lotteries is a linear combination of $1$ and $\theta_{(1)}+\theta_{(2)}^3$, 
hence either increasing in both coordinates of $\theta$ or decreasing in both; so SCD holds. But the types that strictly prefer the third outcome to the first are those with $\theta_{(1)}+\theta_{(2)}^3>0$, a set whose boundary is a curve rather than a line. Hence DSCD fails, which rules out an affine representation; nor are preferences one-dimensional, as every line in $\Reals^2$ contains types that strictly prefer the third outcome to the first and also types with the reverse strict preference.
\end{example}

\autoref{prop:DSCDstar} has three assumptions: $\Theta=\Reals^n$; $u(a,\theta)$ is differentiable in $\theta$; and no type is totally indifferent (or the weaker version in \autoref{fn:totalindiff}). \autoref{eg:spiraling} below shows that the first assumption cannot be dropped; \autoref{eg:nondifferentiable} shows that the latter two cannot be jointly dropped.

\begin{example}
    \label{eg:spiraling}
    Let $\Theta=\R^2_{>0}$, $X=\{x_0,x_1,x_2\}$, and 
    $A=\Delta X$. Consider expected utility preferences $u:A\times \Theta \to \Reals$ induced by the vNM utility $\bernoulli:X\times \Theta \to \Reals$ given by
        \begin{align*}
        \bernoulli(x,\theta)=\begin{cases}
            0 &\text{ if } x=x_0\\
            1 &\text{ if } x=x_1\\
            \tan^{-1}\left(\frac{\theta_{(1)}}{\theta_{(2)}}\right) &\text{ if } x=x_2.
        \end{cases}
    \end{align*} 
So $\bernoulli(x_0,\theta)$ and $\bernoulli(x_1,\theta)$ are constant in $\theta$, and, for any $\theta$, there is a hyperplane through $\theta$ such that $\bernoulli(x_2,\theta')> \bernoulli(x_2,\theta)$ for $\theta'$ on one side of the hyperplane, $\bernoulli(x_2,\theta')< \bernoulli(x_2,\theta)$ for $\theta'$ on the other side, and equality on the hyperplane. 
It follows that for any $a,a'\in \Delta X$,
$\sum_{i=0}^2 [a(x_i)-a'(x_i)]\bernoulli(x_i,\cdot)$
is DSC, and so there is DSCD over lotteries. 

However, it can be checked that neither is $u$ one-dimensional nor does it have an affine representation.
\end{example}

\begin{example}
    \label{eg:nondifferentiable}
    Let $\Theta=\R^2$, $X=\{x_0,x_1,x_2\}$, and 
    $A=\Delta X$. Consider expected utility preferences $u:A\times \Theta \to \Reals$ induced by the vNM utility $\bernoulli:X\times \Theta \to \Reals$ given by
    \begin{align*}
        \bernoulli(x,\theta)=\begin{cases}
            0 &\text{ if } x=x_0\\
            \theta_{(1)}+\theta_{(2)} &\text{ if } x=x_1\\
    \sign\left(\theta_{(2)}\right)e^{\left(\frac{-\left|\theta_{(1)}\right|}{\left|\theta_{(2)}\right|}\right)}\sqrt{\theta_{(1)}^{2}+\theta_{(2)}^{2}} &\text{ if } x=x_2.
        \end{cases}
    \end{align*} 
The expression for $\bernoulli(x_2,\theta)$ is defined to be zero when $\theta_{(2)}=0$, which is its continuous extension.
\autoref{fig:nondifferentiable} graphs $\bernoulli(x_1,\cdot)$ and $\bernoulli(x_2,\cdot)$. Note that $\bernoulli(x_2,\cdot)$ is not differentiable and type $\theta=0$ is totally indifferent. It can be verified that there is DSCD over lotteries,
    but neither is $u$ one-dimensional nor does it have an affine representation.
    
    We note that both \autoref{eg:spiraling} and \autoref{eg:nondifferentiable} answer a question posed by \citet[][p.~323]{Grandmont:78} on whether there are multidimensional preferences satisfying his conditions (H.1) and (H.2) that do not have an affine representation: there are.
    \begin{figure}
        \centering
        \includegraphics[width=0.4\textwidth]{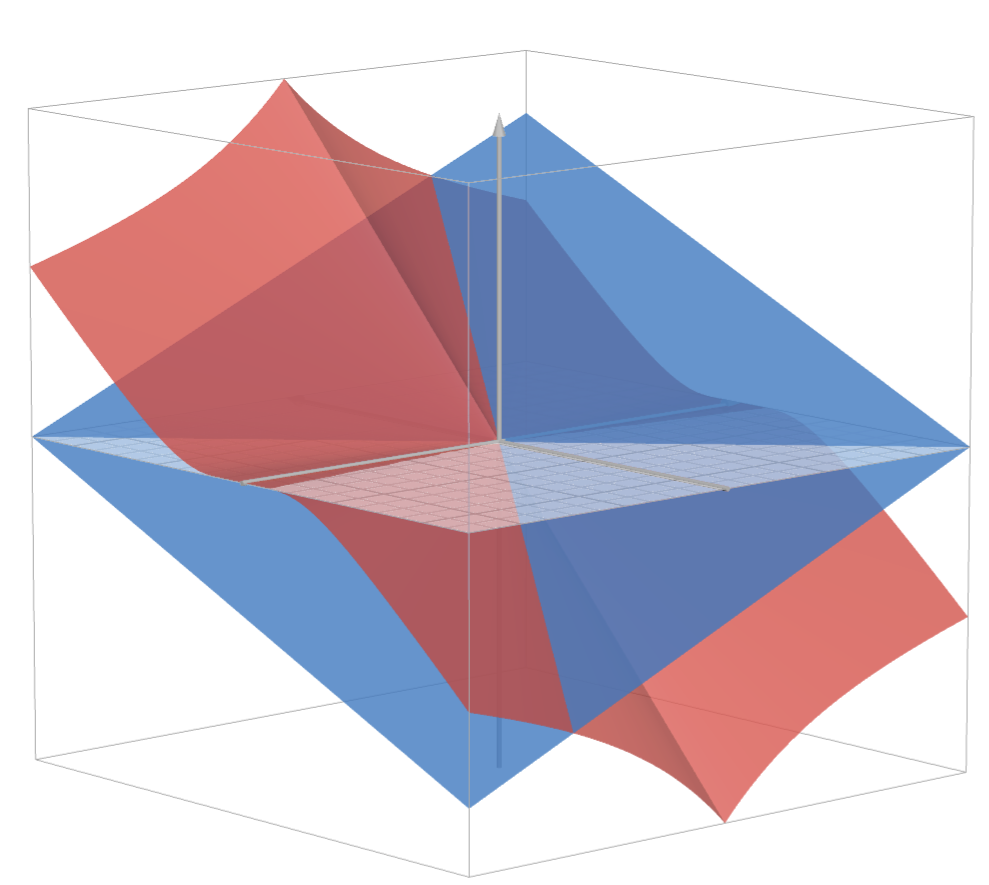}
        \caption{The utilities in \autoref{eg:nondifferentiable}, with $\bernoulli(x_1,\theta)$ in blue and $\bernoulli(x_2,\theta)$ in red.}
        \label{fig:nondifferentiable}
    \end{figure}
\end{example}

\pagebreak

\section{Supplementary Appendix}
\label{sec:suppapp}

\subsection{On \citepos{Grandmont:78} Characterization}
\label{sec:Grandmont}

In this subsection, we elaborate on the connection between \autoref{prop:CC-DSCD} and \cite{Grandmont:78}. It is useful to begin with a strengthening of DSCD.

 \begin{definition} \label{def:SDSC} A function $f:\Theta\rightarrow \R$ 
is \emph{strictly directionally single crossing (strict DSC)} if it is DSC in some direction $\alpha$ and, in addition, either $f(\cdot)=0$ or
for all $\theta,\theta'\in \Theta$,
	\[\Big [ (\theta-\theta')\cdot \alpha > 0 \text{ and } f(\theta')=0 \Big ] \implies f(\theta)\neq 0.\] 
The utility function $u:A\times \Theta \to \Reals$ has \emph{strict directionally single-crossing differences} (strict DSCD) if for all $a,a' \in A$, the difference $u(a,\theta)-u(a',\theta)$ is strictly DSC.
\end{definition}

Strict DSC can be interpreted as requiring that there is a hyperplane such that the sets $%
\{\theta:f(\theta)<0\}$ and $%
\{\theta:f(\theta)>0\}$ are either empty, full, 
or intersections of $\Theta$ with half-spaces defined by that hyperplane, and $\{\theta:f(\theta)=0\}$ is either empty, full, or the intersection of $\Theta$ with that hyperplane. So \autoref{fig:DSC1} satisfies strict DSC whereas \autoref{fig:DSC2} violates it. 

As explained subsequently, the following result follows from \citet{Grandmont:78}.
\begin{proposition}
\label{prop:SDSCD}
    Assume $\Theta$ is open and that 
    \begin{align}
        \label{e:continuity}
        \forall a,a' \in A,
        \text{ the set $\{\theta: u(a,\theta)\geq u(a',\theta)\}$ is closed in
        $\Theta$.}
    \end{align}    
    Then $u$ has convex choice and segment-thin indifferences if and only if it has strict DSCD.
\end{proposition}

Hence, under \autoref{prop:SDSCD}'s two assumptions, convex choice and segment-thin indifferences are fully characterized by strict DSCD.\footnote{\label{fn:strictDSCDassms}Strict DSCD implies convex choice and segment-thin indifferences without either assumption. But absent either assumption, strict DSCD can fail despite convex choice and segment-thin indifferences. Here is an example in which $\Theta$ is not open: $\Theta=[0,1]$, $A=\{a',a''\}$, $u(a',\cdot)=0$, and $u(a'',\theta)=0$ if $\theta\in \{0,1\}$ and $u(a'',\theta)>0$ otherwise. For an example absent \eqref{e:continuity}, consider the left figure in \autoref{eg:strictviolation}.} By comparison, \autoref{prop:CC-DSCD} required neither assumption and showed that convex choice alone is ``almost'' characterized by just DSCD.

To tie \autoref{prop:SDSCD} to \citet{Grandmont:78}, say that $u$ has \emph{convex weak choice} if for all $a,a'\in A$, the set $\{\theta:u(a,\theta)\geq u(a',\theta)\}$ is convex.  In our terminology, \citeauthor{Grandmont:78}'s condition (H.2) is the conjunction of convex weak choice, convex choice, and segment-thin indifferences.\footnote{\label{fn:convexweakchoice}The one-dimensional example in \autoref{fn:strictDSCDassms} has convex choice and segment-thin indifferences but not convex weak choice. Note that given convex choice and segment-thin indifferences, convex weak choice is equivalent to convex indifference, i.e., for all $a,a'$, the set $\{\theta:u(a,\theta)=u(a',\theta)\}$ is convex.} \citeauthor{Grandmont:78}'s condition (H.1) is our \eqref{e:continuity}. His result can thus be stated as:

\begin{proposition}[\cite{Grandmont:78}]
\label{prop:Grandmont}
    Assume $\Theta$ is open.
    Then $u$ satisfies \eqref{e:continuity} and has convex weak choice, convex choice, and segment-thin indifferences if and only if for all $a,a'\in A$, either
    \begin{enumerate}
        \item \label{Grandmont:a} $\{\theta:u(a',\theta)>u(a,\theta)\}=\Theta$ or $\{\theta:u(a',\theta)<u(a,\theta)\}=\Theta$ or \mbox{$\{\theta:u(a',\theta)=u(a,\theta)\}=\Theta$}; or 
        \item \label{Grandmont:b}there is $\alpha \in \Reals^n\setminus 0$ and $c\in \Reals$ such that $\alpha \cdot \theta>c$ if $u(a,\theta)>u(a',\theta)$, $\alpha \cdot \theta<c$ if $u(a,\theta)<u(a',\theta)$, and $\alpha \cdot \theta=c$ if $u(a,\theta)=u(a',\theta)$.
    \end{enumerate}
\end{proposition}

The ``if'' directions of both \autoref{prop:SDSCD} and \autoref{prop:Grandmont} are straightforward, so let us explain how each proposition's ``only if'' can be obtained from the other. To go from \autoref{prop:Grandmont} to \autoref{prop:SDSCD}, we can first observe that convex choice and segment-thin indifferences imply convex weak choice when $\Theta$ is open,\footnote{Suppose that,  contrary to convex weak choice, $u(a,\theta_1)\geq u(a',\theta_1)$ and $u(a,\theta_3)\geq u(a',\theta_3)$ but $u(a,\theta_2)<u(a',\theta_2)$ for some $\theta_2\in \ell(\theta_1,\theta_3)$. If $\Theta$ is open, then there are $\theta_0$ and $\theta_4$ such that $\theta_1,\theta_3\in \ell(\theta_0,\theta_4)$. If $u(a,\theta)<u(a',\theta)$ for either $\theta\in \{\theta_0,\theta_4\}$, then convex choice fails. If $u(a,\theta)=u(a',\theta)$ for either $\theta\in \{\theta_0,\theta_4\}$, then either segment-thin indifferences or convex choice fails.
But if $u(a,\theta)>u(a',\theta)$ for both $\theta\in \{\theta_0,\theta_4\}$, then convex choice fails.} and then observe that the consequent of \autoref{prop:Grandmont} implies strict DSCD. To go from \autoref{prop:SDSCD} to \autoref{prop:Grandmont}, it suffices to observe that given \eqref{e:continuity}, strict DSCD implies the consequent of \autoref{prop:Grandmont}.\footnote{Suppose $u$ has strict DSCD in direction $\alpha\neq 0$ and the sign of $u(a,\cdot)-u(a',\cdot)$ is not constant (otherwise, point \ref{Grandmont:a} of \autoref{prop:Grandmont} follows immediately). Then, \eqref{e:continuity} implies that there is $\theta'$ with $u(a,\theta')=u(a',\theta')$. Consider the hyperplane in the direction of $\alpha$ passing through $\theta'$. Strict DSCD implies that 
all types on that hyperplane are indifferent between $a$ and $a'$, all types to the left (i.e., types $\theta$ with $\alpha \cdot \theta < \alpha \cdot \theta'$) strictly prefer $a'$, and all types to the right strictly prefer $a$. Point \ref{Grandmont:b} of  \autoref{prop:Grandmont} follows, with $c=\alpha \cdot \theta'$.}

\subsection{On \citepos{MM:88} Single Crossing}
\label{sec:MM}

In this subsection, we construct a one-dimensional example in which the utility function satisfies \citepos{MM:88} generalized single crossing  (GSC) but does not have convex choice.
We construct a mechanism that is locally IC but not IC, as predicted by \autoref{prop:CC-IC}.

Let $\Theta=[0,1]$, $A=[0,1]$, and $\kappa>0$.\footnote{\citet{MM:88} allow for transfers. We implicitly set transfers to zero to simplify.} Preferences are given by
\begin{align*}
    u(a,\theta)= \begin{cases}
        -\kappa a -2\left(a-\frac12\right)^4 \int_0^{\theta} (2-s) \,\mathrm ds  & \text{ if } a\le \frac{1}{2}\\
        -\kappa a -2\left(a-\frac12\right)^4 \int_0^{\theta} (1+s) \,\mathrm ds  & \text{ if } a> \frac{1}{2}.
    \end{cases}
\end{align*}
Below, we use subscripts on $u$ to denote partial derivatives.

\paragraph{Generalized single crossing.} To verify that $u$ satisfies \citeauthor{MM:88}'s GSC, we must show that for all $a,\theta,\theta'$ there is $\lambda>0$ such that 
\begin{equation}
    u_a(a,\theta)-u_a(a,\theta')  = \lambda u_{a\theta}(a,\theta') (\theta-\theta').
    \label{e:GSC}
\end{equation}

Differentiation yields
\begin{align*}
    u_a(a,\theta)= \begin{cases}
        -\kappa -8(a-\tfrac12)^3 \int_0^{\theta} (2-s) \,\mathrm ds  & \text{ if } a\le \frac{1}{2}\\
        -\kappa -8(a-\tfrac12)^3 \int_0^{\theta} (1+s) \,\mathrm ds  & \text{ if } a > \frac{1}{2},
    \end{cases}
\end{align*}
and
\begin{align*}
    u_{a\theta}(a,\theta)= \begin{cases}
         -8(a-\tfrac12)^3 (2-\theta)   & \text{ if } a\le \frac{1}{2}\\
         -8(a-\tfrac12)^3 (1+\theta) & \text{ if } a > \frac{1}{2}.
    \end{cases}
\end{align*}

If $u_{a\theta}(a,\theta)>0$ then $a<\frac{1}{2}$, hence
\[ u_a(a,\theta)-u_a(a,\theta') = -8(a-\tfrac12)^3 \int_{\theta'}^{\theta} (2-s)\,\mathrm ds \]
 is strictly positive (negative) if and only if $\theta'<(>)\theta$.  Similarly, if $u_{a\theta}(a,\theta)<0$ then $a>\frac{1}{2}$, hence
\[ u_a(a,\theta)-u_a(a,\theta') = -8(a-\tfrac12)^3 \int_{\theta'}^{\theta} (1+s)\,\mathrm ds \]
 is strictly positive (negative) if and only if $\theta'>(<)\theta$. 
Lastly, if $u_{a\theta}(a,\theta)=0$ then $a=\frac{1}{2}$ and $u_a(a,\theta)-u_a(a,\theta')=0$.
Therefore, for all $a,\theta,\theta'$ there is $\lambda>0$ satisfying \autoref{e:GSC}.

\paragraph{No convex choice and the insufficiency of local IC.} Let $\kappa>0$ be small and compare $a=0$ with $a=1$: 
\begin{align*}
 u(0,\theta)&=
 -2\left(\frac12\right)^4 \int_0^{\theta} (2-s) \,\mathrm ds,  \\
 u(1,\theta)&=  -\kappa
 -2\left(\frac12\right)^4 \int_0^{\theta} (1+s) \,\mathrm ds.
\end{align*}

Hence,
\begin{align*}
8[u(1,\theta)-%
u(0,\theta)]&=\int_0^{\theta}\left[ (2-s)-(1+s)\right]\mathrm ds - 8\kappa\\
&=\int_0^{\theta} (1-2s) \,\mathrm ds - 8\kappa\\
&=\theta-\theta^2 - 8\kappa,
\end{align*}
with roots $$\underline \theta = \frac{1 - \sqrt{1-32\kappa}}{2} \text{ and } \overline \theta = \frac{1 + \sqrt{1-32\kappa}}{2}.$$
These two types are indifferent between the actions $a=0$ and $a=1$, types in $\left(\underline \theta ,\overline \theta\right)$ strictly prefer $a=1$, and  types 
outside $\left[\underline \theta,\overline \theta\right]$
strictly prefer $a=0$. So convex choice---and hence also DSCD---fails.
Moreover, the following mechanism is locally IC but not IC: 
\begin{align*}
    m(\theta)=\begin{cases}
        0 \text{ if } \theta \leq \underline \theta\\
        1 \text{ otherwise.}
    \end{cases}
\end{align*}

\end{document}